\documentclass[a4paper,onecolumn,11pt,accepted=2024-01-30]{quantumarticle}
\pdfoutput=1
\usepackage[utf8]{inputenc}
\usepackage[english]{babel}
\usepackage[T1]{fontenc}
\usepackage{amsthm}
\usepackage{amsmath}
\usepackage{latexsym}
\usepackage{amsfonts}
\usepackage{amssymb}
\usepackage{color}
\usepackage{bbm,dsfont}
\usepackage{graphicx}
\usepackage{subfigure}
\usepackage{mathrsfs}
\usepackage{hyperref}
\usepackage{enumerate}
\usepackage{bbding}

\usepackage[normalem]{ulem} 

\newtheorem{proposition}{Proposition}
\newtheorem*{proposition*}{Proposition}

\newtheorem*{theorem*}{Theorem}
\newtheorem{lemma}{Lemma}
\newtheorem{corollary}{Corollary}

\theoremstyle{definition}

\newtheorem{remark}{Remark}
\newtheorem{example}{Example}
\newtheorem{definition}{Definition}
\newtheorem*{definition*}{Definition}

\newcommand{\complex}{\mathbb C} 
\newcommand{\nat}{\mathbb N} 

\newcommand{\hi}{\mathcal{H}} 
\newcommand{\hik}{\mathcal{K}} 
\newcommand{\hv}{\mathcal{V}} 
\newcommand{\lh}{\mathcal{L(H)}} 
\newcommand{\lsh}{\mathcal{L}_s(\mathcal{H})} 
\newcommand{\lk}{\mathcal{L(K)}} 
\newcommand{\sh}{\mathcal{S(H)}} 
\newcommand{\shik}{\mathcal{S(K)}} 
\newcommand{\shv}{\mathcal{S(V)}} 
\newcommand{\eh}{\mathcal{E(H)}} 
\newcommand{\ip}[2]{\left\langle\,#1\,|\,#2\,\right\rangle} 
\newcommand{\ket}[1]{|#1\rangle} 
\newcommand{\bra}[1]{\langle#1|} 
\newcommand{\kb}[2]{|#1\rangle\langle#2|} 
\newcommand{\tr}[1]{\mathrm{tr}\left[#1\right]} 
\newcommand{\ptr}[2]{\mathrm{tr}_{#1}\left[#2\right]} 

\newcommand{\obs}{\mathcal{O}}
\newcommand{\A}{\mathsf{A}}
\newcommand{\B}{\mathsf{B}}
\newcommand{\C}{\mathsf{C}}
\newcommand{\Eo}{\mathsf{E}}
\newcommand{\Fo}{\mathsf{F}}
\newcommand{\Go}{\mathsf{G}}
\newcommand{\Po}{\mathsf{P}}
\newcommand{\X}{\mathsf{X}}
\newcommand{\Z}{\mathsf{Z}}

\newcommand{\N}{\mathcal{N}}

\newcommand{\E}{\mathcal{E}}
\newcommand{\F}{\mathcal{F}}
\newcommand{\Q}{\mathcal{Q}}
\newcommand{\R}{\mathcal{R}}
\newcommand{\I}{\mathcal{I}}

\newcommand{\J}{\mathcal{J}}
\newcommand{\G}{\mathcal{G}}
\newcommand{\ins}{\mathrm{Ins}}

\newcommand{\Y}{\mathcal{Y}}

\newcommand{\ch}{\mathrm{Ch}}

\newcommand*{\comp}{\hbox{\hskip0.85mm$\circ\hskip-1mm\circ$\hskip0.85mm}}
\newcommand{\nd}{\ {\scriptstyle{\mathsf{ND}}} \ }

\begin{document}

\title[]{Incompatibility of quantum instruments}

\author{Leevi Lepp\"aj\"arvi}
\affiliation{RCQI, Institute of Physics, Slovak Academy of Sciences, Dúbravská cesta 9, 84511 Bratislava, Slovakia }
\email{leevi.leppajarvi@savba.sk}
\orcid{0000-0002-9528-1583}

\author{Michal Sedl\'ak}
\affiliation{RCQI, Institute of Physics, Slovak Academy of Sciences, Dúbravská cesta 9, 84511 Bratislava, Slovakia }
\affiliation{Faculty of Informatics, Masaryk University, Botanick\'{a} 68a, 602\,00 Brno, Czech Republic }
\email{michal.sedlak@savba.sk}
\orcid{0000-0002-9342-5862}

\begin{abstract}
Quantum instruments describe outcome probability as well as state change induced by measurement of a quantum system.  Incompatibility of two instruments, i. e. the impossibility to realize them simultaneously on a given quantum system, generalizes incompatibility of channels and incompatibility of positive operator-valued measures (POVMs). We derive implications of instrument compatibility for the induced POVMs and channels. We also study relation of instrument compatibility to the concept of non-disturbance. Finally, we prove equivalence between instrument compatibility and postprocessing of certain instruments, which we term complementary instruments. We illustrate our findings on examples of various classes of instruments.
\end{abstract}

\maketitle
\section{Introduction} \label{sec:1}
\noindent Incompatibility of measurements captures the fact that not all (and not even all pairs of) quantum measurements can be measured jointly simultaneously, and  it is widely recognized as one of the most important nonclassical features of quantum theory. The roots of incompatibility can be found already in the works of Heisenberg \cite{Heisenberg27} and Bohr \cite{Bohr28} with the most paradigmatic example being the inability to sharply measure the position and momentum of a particle at the same time. 

Once the concept of incompatibility was recognized, it was first characterized through commutativity relations of sharp observables and later generalized to existence of a joint measurement device with suitable marginals to encompass the modern description of quantum measurements via the positive operator-valued measures (POVMs) (see e.g. \cite{GuhneHaapasaloKraftPellonpaaUola21} for a short historical review). Indeed various works linked incompatibility of POVMs to Bell nonlocality (as Bell inequalities can be violated only using incompatible measurements) \cite{Fine82,WolfPerez-GarciaFernandez09}, contextuality \cite{LiangSpekkensWiseman11,XuCabello19,TavakoliUola20}, steering \cite{QuintinoVertesiBrunner14}, various quantum information tasks such as state discrimination \cite{SkrzypczykSupicCavalcanti19, CarmeliHeinosaariToigo19, UolaKraftShangYuGuhne19} and random access codes \cite{CarmeliHeinosaariToigo20,HeinosaariLeppajarvi22} and in general to the nonclassicality of an operational theory \cite{Plavala16}. For more detailed review of incompatibility we encourage the reader to see  \cite{GuhneHaapasaloKraftPellonpaaUola21,HeinosaariMiyaderaZiman16}.

The concept of joint measurability is an operational notion involving any preparation, transformation or measurement devices with various types of inputs and outputs, so it is not limited to POVMs. Indeed, the (in)compatibility of quantum channels, i.e., devices that describe the transformations between quantum systems, was introduced in \cite{HeinosaariMiyadera17} and later studied e.g. in \cite{Haapasalo19, Kuramochi18, GirardPlavalaSikora21}. More generally, the (in)compatibility of channels acting between any two systems (classical, quantum or mixed quantum-classical) has been considered in \cite{CarmeliHeinosaariMiyaderaToigo19}. In particular, the compatibility of quantum instruments, i.e., devices that can be used to measure the quantum system while also including the description of the post-measurement state thus allowing for sequential measurements, was explicitly studied recently in  \cite{MitraFarkas22} and it has resulted into an increased interest and a series of works \cite{JiChitambar21,DArianoPerinottiTosini22,MitraFarkas22b,BuscemiKobayashiMinagawaPerinottiTosini22} on the topic.

In particular, in \cite{MitraFarkas22} the authors proposed two possible definitions for compatibility of instruments: the traditional compatibility, which strictly generalizes the concept from the compatibility of POVMs as finding a joint instrument with proper classical marginals, and the parallel compatibility, which captures the idea that the joint device should be able to produce all the outputs, including the post-measurement states, simultaneously for all original instruments. Advocating for the latter definition, the same authors focused on quantitative characterization of incompatibility of instruments through incompatibility robustness in \cite{MitraFarkas22b}. On the other hand, the authors of \cite{JiChitambar21} focused on the traditional incompatibility and  considered it as a resource for programmable quantum instruments. In \cite{BuscemiKobayashiMinagawaPerinottiTosini22} the authors introduced more different notions for compatibility of instruments, including the previous two, and used them to obtain hierarchy of resource theories for incompatibility, and in \cite{DArianoPerinottiTosini22} the authors considered incompatibility of measurement devices in an even more general setting of operational probabilistic theories and connected it to disturbance caused by measurements.

\begin{figure}
\centering
\includegraphics[scale=0.25]{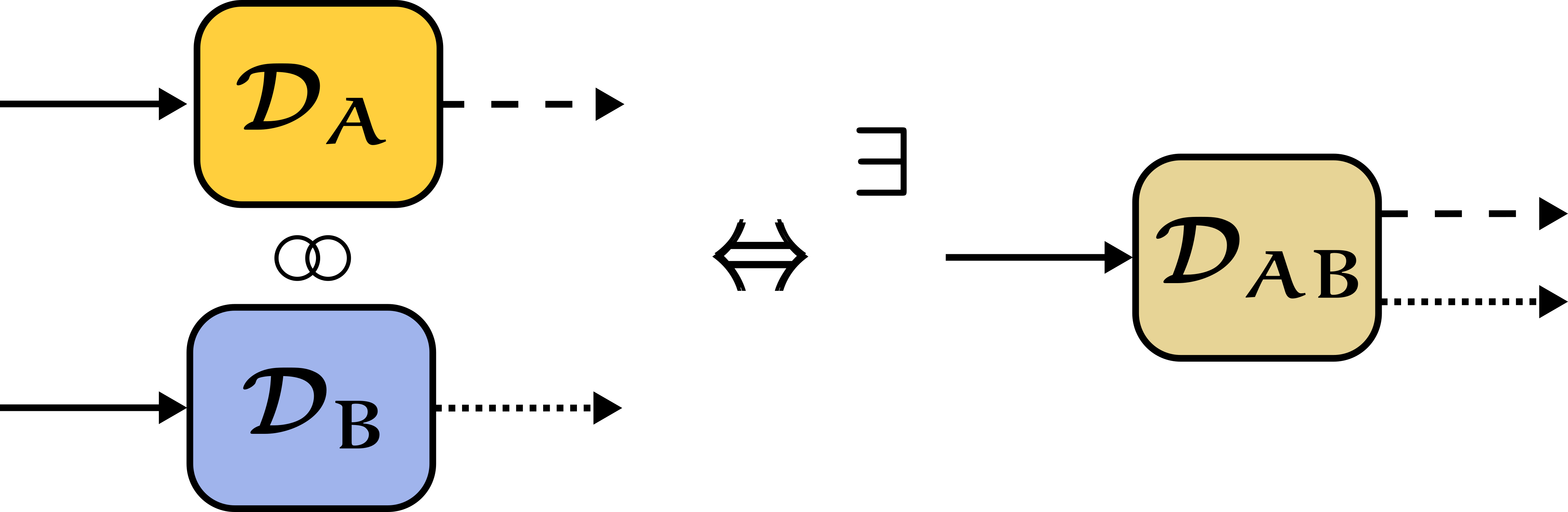}
\caption{\label{fig:comp} Input-output devices $\mathcal{D}_A$ and $\mathcal{D}_B$ are compatible, denoted by $\mathcal{D}_A \comp \mathcal{D}_B$, if and only if there exists a joint device $\mathcal{D}_{AB}$ which produces an output for both original devices simultaneously with a single input in each use of the joint device.}
\end{figure}

The starting point for our work is what one can view as an operational description of parallel compatibility: compatible devices should allow for joint simultaneous implementation meaning that a joint device with a single input should be used to obtain the correct output for each original device at each experimental run. For general input-output devices this is illustrated in Fig. \ref{fig:comp}. We note that the notion of parallel compatibility can be applied to any set of possibly different types of devices (in any operational theory) as long as they have the same input types, thus unifying the existing notions of POVM-POVM, POVM-channel and channel-channel compatibility. Conceptually parallel compatibility was considered in \cite{HeinosaariMiyaderaZiman16} and mathematically the resulting definitions can be obtained from the channel compatibility considered in \cite{CarmeliHeinosaariMiyaderaToigo19}. The terminology is adopted from \cite{MitraFarkas22}. Thus, in this work we consider compatible instruments to be parallelly compatible. 

We study in detail the relation between the compatibility of different devices, namely measurements (described by POVMs), channels and instruments, and show that indeed the compatibility of instruments generalizes all the previously defined notions. Indeed, although instruments include a description both for the induced POVM and the state change induced by the measurement, studying the compatibility relations between the induced POVMs/channels shows that compatibility of instruments cannot be reduced into the compatibility of its induced POVMs or channels. Next we generalize the concept of non-disturbance to quantum instruments and generalize the known connections between compatibility and non-disturbance (see e.g. \cite{HeinosaariWolf10}) to instruments: we show that if the measurement process of a POVM does not disturb an instrument, then the POVM and the instrument must be compatible. Lastly, by introducing the concept of a complementary instrument we can show that the instrument postprocessing relation \cite{LeppajarviSedlak21} can be used to characterize compatible instruments, thus strictly generalizing the result of \cite{HeinosaariMiyadera17} for channels: two instruments are compatible if and only if one of them can be postprocessed from the complementary instrument of the other. While in literature we found almost no case studies of compatible channel type pairs originating from postprocessing of a complementary channel, it was possible to provide some simple illustrative examples in the general case of instruments. This shows that this characterization can be quite useful for certain simple classes of instruments.

The structure of the manuscript is as follows: In Sec. \ref{sec:main-results} we shortly summarize the main results that we derive in this work. In Sec.  \ref{sec:2} we introduce all the preliminaries and the used notation. In Sec. \ref{sec:3} we define compatibility of different devices and study their relations. Sec. \ref{sec:4} focuses on exploring the connection of compatibility and non-disturbance. In Sec. \ref{sec:5} we show the connection of compatibility and postprocessing of instruments and give a useful characterization for compatible instruments. Lastly, in Sec. \ref{sec:6} we use the previously obtained characterization to exemplify compatibility between simple classes of instruments. In Sec. \ref{sec:7} we give our summary and conclusions.

\section{Main results} \label{sec:main-results}

The aim of this paper is to investigate incompatibility of instruments and to find links of this problem to other problems and known results. Formally, a quantum instrument $\I$ with a finite set of measurement outcomes $\Omega$ that describes the measurement and the related state transformation process between Hilbert spaces $\hi$ and $\hik$ is a mapping $\I: x \mapsto \I_x$ from $\Omega$ to the set of completely positive linear maps from $\hi$ to $\hik$ such that $\Phi^{\I}:=\sum_{x \in \Omega} \I_x $ is a quantum channel, i.e., a completely positive trace-preserving map. For any quantum state (density matrix) $\varrho $, the (unnormalized) conditional post-measurement state is $\I_x(\varrho)$, where $x$ is the outcome obtained in the measurement. Probability of obtaining $x$ is governed by the induced POVM $\A^\I $, i.e., a mapping $\A^\I: x \mapsto \A^\I(x)$ from $\Omega$ to the set of positive operators on $\hi$ such that $\sum_{x \in \Omega} \A^\I(x) = I$, defined via the formula $\tr{\A^\I(x)\varrho} = \tr{\I_x(\varrho)}$. The set of instruments from $\hi$ to $\hik$ with outcome set $\Omega$ is denoted by $\ins(\Omega,\hi,\hik)$. 

Quantum instrument can be seen as a concept that easily encompasses the concept of quantum channel or a POVM as its special case if the outcome set of the instrument has just a single element or the output space is one-dimensional, respectively. With suitable definitions this allows us to consider compatibility of a POVM with some other POVM or a channel or an instrument, or compatibility of a channel with another channel or an instrument as a special case of compatibility of two instruments. We will make use of this fact to avoid repeating the arguments separately for POVMs or  channels,  although we might state the results separately for each of them for readers convenience. The above unifying notion of compatibility, as discussed in detail in Sec. \ref{sec:3}, takes the following form known as parallel compatibility \cite{MitraFarkas22}:

\begin{definition*}[Compatibility]
Two instruments $\I \in \ins(\Omega,\hi,\hik)$ and $\J \in \ins(\Lambda, \hi,\hv)$ are compatible, denoted by $\I \comp \J$, if there exists a joint instrument $\G \in \ins(\Omega \times \Lambda, \hi, \hik \otimes \hv)$ such that
\begin{align*}
\sum_{x \in \Omega} \ptr{\hik}{\G_{(x,y)}(\varrho)} = \J_y(\varrho) \quad \forall y \in \Lambda, \quad \sum_{y \in \Lambda} \ptr{\hv}{\G_{(x,y)}(\varrho)} = \I_x(\varrho) \quad \forall x \in \Omega
\end{align*}
for all states $\varrho$. 
\end{definition*}

\begin{remark}[Parallel vs. traditional compatibility]
    As we already mentioned, the previous definition is called \emph{parallel compatiblity} in \cite{MitraFarkas22} in contrast to the other notion of compatibility for instruments considered by the authors in \cite{MitraFarkas22}, which they call \emph{traditional compatibility}: \emph{two instruments $\I \in \ins(\Omega,\hi,\hik)$ and $\J \in \ins(\Lambda, \hi,\hik)$ are traditionally compatible if there exists an instrument $\G \in \ins(\Omega \times \Lambda, \hi, \hik )$ such that $\sum_{x \in \Omega} \G_{(x,y)} = \J_y$ for all $y \in \Lambda$ and $ \sum_{y \in \Lambda} \G_{(x,y)} = \I_x$ for all $x \in \Omega$.}

    Thus, in order for two instruments to be traditionally compatible, the instruments should have also the same output spaces and the joint device could be used to recover one of the original instruments just by ignoring the classical measurement outcome associated to the other instrument and vice versa. Furthermore, what follows from the definition is that for two instruments to be traditionally compatible it is necessary that they have the same induced channel. Next, we want to point out some differences between traditional and parallel compatibility.
    
    While it is clear that traditional compatibility recovers the compatibility of POVMs, it does not follow the spirit of the operational description of parallel compatibility of general (possibly different type) input-output devices. Naturally, traditional compatibility is also an operational notion, but it is also a more restrictive notion of compatibility since it can only be defined for devices with the same output as well as the same input. Thus, compatibility of a POVM and a channel or a POVM and an instrument cannot even be considered to be traditionally compatible. On the other hand, parallel compatibility of a POVM and another POVM or a channel or an instrument can be linked to an important physical problem of non-distrubance (see Sec. \ref{sec:4} for details).
    
    Furthermore, even though the definition of traditional compatibility can be applied for channels (with the same output space), one can immediately see that in this case the definition becomes trivial, since a channel can be traditionally compatible only with itself. On the other hand, parallel compatibilty of channels is closely related to tasks such as cloning and broadcasting which can be phrased simply in terms of parallel compatiblility of identity channels. Traditional compatibility is a valid concept for instruments, which should be studied on its own, including exploration of its connections to other information-theoretic tasks. However, it does not serve well as a unifying concept of compatibility for general input-output devices; a concept that we are trying to pursue in this work. 

    Another difference between traditional compatibility and parallel compatibility is that even though both of them generalize compatibility of POVMs, parallel compatibility respects wider class of reductions. What we mean by that is that if one considers POVMs as measure-and-prepare channels/instruments, which encode the classical measurement outcomes into (distinguishable) quantum states, it can be shown (see Sec. \ref{sec:m-a-p} for details) that two POVMs are compatible if and only if the associated measure-and-prepare channels/instruments are parallelly compatible. On the other hand, one can find examples where compatible POVMs have traditionally incompatible measure-and-prepare channels/instruments. 
\end{remark}

In the current manuscript we carefully analyze consequences of instrument compatibility for the induced POVMs and channels of the two instruments. We also introduce and study a previously omitted concept of compatibility of a POVM and an instrument which is a special case of the previous definition when one considers the POVM as an instrument with a one-dimensional output space. We find that \emph{compatibility with an instrument automatically implies compatibility with its induced POVM and its induced channel}. These necessary conditions might be often easier to check, and would quite often suffice to conclude incompatibility of a given pair of instruments. On the other hand, we show that \emph{instrument compatibility is not equivalent to simultaneous compatibility of the induced POVMs and the induced channels of the instruments}. 

Compatibility of POVMs is known to be closely related to the important operational concept of non-disturbance.
\begin{definition*}[Non-disturbance]
An instrument $\I$ is said \emph{not to disturb an instrument $\J$} if $\J_y \circ \Phi^\I = \J_y$ for all $y$. We also say that a POVM $\A$ \emph{does not disturb an instrument} $\J $ if there exists an instruments $\I $ with $\A^\I= \A$ such that $\I$ does not disturb $\J$.
\end{definition*}
This concept has been mostly studied in terms of POVMs (second part of the definition when $\J$ is a POVM) but we generalize this concept to general instruments. It is known \cite{HeinosaariWolf10} that if a POVM does not disturb another POVM then the POVMs are compatible, and furhtermore, if one of the POVMs is sharp, i.e., projection-valued, then also the converse holds. We generalize this result in the case when a POVM does not disturb an instrument:

\begin{proposition*}
If a POVM $\A$ does not disturb an instrument $\J$, then $\A$ and $\J$ are compatible. Furthermore, for sharp $\A$ also the converse holds.
\end{proposition*}

Another way how compatibility problem can be approached is by studying constructions, which lead to compatible pairs. For this purpose it is natural to consider how an instrument can be postprocessed in the most general way, because application of such postprocessings will not change compatibility of a pair of instruments. In particular, following \cite{LeppajarviSedlak21}, we say that an instrument $\I$ is a \emph{postprocessing} of an instrument $\J$, denoted by $\I \preceq \J$, if there exist instruments $\{\R^{(y)} \}_{y }$ such that $\I_x = \sum_{y \in \Lambda} \R^{(y)}_x \circ \J_y$ for all $x $.

Previously it was known \cite{HeinosaariMiyadera17} that a pair of channels are compatible if and only if one of them can be postprocessed from the complementary channel of the other channel. Inspired by a notion of complementary channel we defined analogous notion of a complementary instrument. In partucular, for an instrument $\I \in \ins(\Omega, \hi, \hik)$ with a dilation $(\hi_A, W, \Eo)$, meaning that $\I_x(\varrho)=\ptr{\hi_A}{W\varrho W^* \left(\Eo(x) \otimes I_{\hik} \right)}$, we define \emph{complementary instrument} $\I^C \in \ins(\Omega, \hi, \hi_A)$ relative to $(\hi_A, W, \Eo)$ via the formula $\I^C_x(\varrho)=\ptr{\hik}{\left(\sqrt{\Eo(x)}\otimes I_{\hik} \right) W\varrho W^*  \left(\sqrt{\Eo(x)}\otimes I_{\hik}\right) }$.

One of the main results of this paper is that we were able to show that two instruments are compatible if and only if one of them can be postproccessed from the complementary instrument of the other one. 
\begin{theorem*}
Let $\I$ and $\J $ be instruments and let $\I^C$ be a complementary instrument of $\I$ related to any dilation. Then $\I \comp \J$ if and only if $\J \preceq \I^C$.
\end{theorem*}

Thus, we showed that verifying compatibility of a pair of instruments is equivalent to verification of the instrument postprocessing relation between one of the instruments and the complementary instrument of the other. This result can be also seen as a characterization of all instruments that are compatible with a chosen instrument.

Finally, we illustrate the findings presented above by finding particular classes of compatible instruments. For measure-and-prepare instruments, i.e., instruments which are realized by a measurement of some POVM on the input state and preparing an output state based on the obtained measurement outcome, we show that \emph{any instrument is compatible with a measure-and-prepare instrument if and only if it is compatible with its induced POVM}. Similarly, we proved that \emph{a pair of measure-and-prepare instruments are compatible if and only if their induced POVMs are compatible}. For indecomposable instruments, i.e., instruments which are composed only from quantum operations expressible by a single Kraus operator, we show that one of their complementary instruments is always a specific measure-and-prepare instrument with the same induced POVM. Thus, \emph{all instruments compatible with an indecomposable instrument must be an postprocessing of that measure-and-prepare instrument}, and we furthermore characterize the form of these instruments. We also show that \emph{two indecomposable instruments can be compatible only if their induced POVMs are postprocessing equivalent to rank-1 POVMs}. This effectively means that they are both measure-and-prepare instruments which prepare pure states and measure rank-1 postprocessing equivalent POVMs.

\section{Preliminaries and notation} \label{sec:2}
Let $\hi$ be a finite-dimensional complex Hilbert space and let us denote the set of (bounded) linear operators on $\hi$ by $\lh$. We denote the adjoint of an operator $A \in \lh$ by $A^*$ defined as $\ip{\varphi}{A \psi} = \ip{A^* \varphi}{\psi}$ for all $\varphi, \psi \in \hi$, and we denote the set of selfadjoint operators in $\lh$, i.e. operators $A \in \lh$ such that $A^* = A$, by $\lsh$. \emph{States} of a quantum system are represented by the elements of the set of unit-trace positive semi-definite operators on $\hi$, which we denote by $\sh$, so that
\begin{equation*}
\sh= \{ \varrho \in \lsh \, | \, \varrho \geq O, \ \tr{\varrho} =1 \},
\end{equation*}
where $O$ denotes the zero operator on $\hi$.

Selfadjoint operators on $\hi$ bounded by $O$ and $I$, where $I$ (or $I_\hi$ if we want to be more specific) is the identity operator on $\hi$, are called 
\emph{effects} and we denote the set of effect on $\hi$ by $\eh$, i.e.,
\begin{equation*}
\eh = \{E \in \lsh \, | \, O \leq E \leq I \}.
\end{equation*}
The probability that the event corresponding to an effect $E$ is detected in a measurement of a quantum system in a state $\varrho$ is given by the Born rule $\tr{ \varrho E}$. Consequently, an observable $\A$  on $\hi$ with a finite set of outcomes $\Omega$ is described by a \emph{positive operator-valued measure (POVM)}, i.e., by a mapping $\A: x \mapsto \A(x)$ from $\Omega$ to $\eh$ such that $\sum_{x \in \Omega} \A(x) = I$. The set of observables on $\hi$ with outcome set $\Omega$ is denoted by $\obs(\Omega, \hi)$. We say that a POVM $\A \in \obs(\Omega, \hi)$ is \emph{sharp} if it is projection-valued, i.e., $\A(x)^2 = \A(x)$ for all $x \in \Omega$. 

Transformations between sets of states on Hilbert spaces $\hi$ and $\hik$, i.e. between $\sh$ and $\shik$, are mathematically described by \textit{quantum channels}, which are completely positive trace-preserving maps from $\lh$ to $\lk$. The set of channels from $\lh$ to $\lk$ is denoted by $\ch(\hi,\hik)$. We denote the identity channel in $\ch(\hi, \hi)$ by $id_\hi$ or just $id$ if the Hilbert space is evident from the context. Completely positive maps that are not trace-preserving but only trace-nonincreasing are called \textit{quantum operations} and they describe the probabilistic transformations of states. Quantum channels and operations have several useful characterisations some of which we will recall next.

Very well known representation of quantum channels and operations is the \emph{operator-sum form} or the \emph{Kraus decomposition} \cite{KrausBook83}: a linear map $\N: \lh \to \lk$ is a quantum operation if and only if there exists bounded operators (called Kraus operators) $K_i: \hi \to \hik$ for all $i=1,2, \ldots, N $ for some $N \in \nat$ such that $\N(\varrho) = \sum_i K_i \varrho K_i^*$ for all $\varrho \in \lh$ and $\sum_i K_i^* K_i \leq I$. For finite-dimensional $\hi$ and $\hik$ it is possible to choose $\dim(\hi) \dim(\hik)$ or fewer Kraus operators, and the minimal number of Kraus operators for a given operation is called the \textit{Kraus rank} of the operation. Operations/channels with Kraus rank 1 are exactly the \emph{indecomposable} elements in the set of operations/channels meaning that they cannot be written as a non-trivial sum of any other different operations \cite{LeppajarviSedlak21}.

Another important representation of quantum channels and operations is \textit{the Stinespring dilation} \cite{Stinespring55}. If $\N: \lh \to \lk$ is a quantum operation, then its Stinespring dilation is denoted by a tuple $(\hi_A, W)$ of an ancillary Hilbert space $\hi_A$ and a bounded operator $W: \hi \to \hi_A \otimes \hik$ that satisfies $W^* W \leq I$ such that
\begin{equation}\label{eq:stinespring-def}
    \N(\varrho) = \ptr{\hi_A}{W \varrho W^*}
\end{equation}
for all $\varrho \in \lh$. In the case when $\N$ is a quantum channel we have that $W$ is an isometry, i.e., $W^*W=I$. It is an elementary result in Hilbert space operator theory \cite{Stinespring55} that for all quantum operations a Stinespring dilation exists, it is not unique and furthermore that it can be chosen to be \emph{minimal} in the sense that the ancillary space has a minimal dimension, or more formally, that the linear span of the vectors $\{(I \otimes K) W \varphi \, | \, K \in \lk, \varphi \in \hi\} $ is the whole $\hi_A \otimes \hik$. It can be shown that every dilation  $(\hi_A', W')$ is connected to a minimal dilation $(\hi_A, W)$ by an isometry $V: \hi_A \to \hi_A'$ such that $W'=(V \otimes I)W $. We note that here $V$ is unitary if and only if also $(\hi_A', W')$ is minimal.

While a measurement device with an output of only classical measurement outcomes is described by a POVM, a \emph{quantum instrument} is used to describe a measurement scenario where also the post-measurement state is taken into account. Formally, a quantum instrument $\I$ with a finite set of measurement outcomes $\Omega$ that describes the measurement and the related state transformation process between Hilbert spaces $\hi$ and $\hik$ is a mapping $\I: x \mapsto \I_x$ from $\Omega$ to the set of quantum operations from $\lh$ to $\lk$ such that $\Phi^{\I}:=\sum_{x \in \Omega} \I_x  \in \ch(\hi, \hik)$ is a channel. For any  state $\varrho \in \sh$, the (unnormalized) conditional post-measurement state is $\I_x(\varrho)$, where $x$ is the outcome obtained in the measurement. Probability of obtaining $x$ is governed by the \emph{induced POVM} $\A^\I \in \obs(\Omega, \hi)$ via the formula $\tr{\A^\I(x)\varrho} = \tr{\I_x(\varrho)}$. The set of instruments from $\lh$ to $\lk$ with outcome set $\Omega$ is denoted by $\ins(\Omega, \hi, \hik)$. In the case when the input and the output spaces are the same, $\hik = \hi$, we denote the set simply $\ins(\Omega, \hi)$. We also call an instrument \emph{indecomposable} if each of its (nonzero) operations is indecomposable.

\begin{example}\label{ex:luders}
For a POVM $\A \in \obs(\Omega, \hi)$ we define the \emph{Lüders instrument} $\I^\A \in \ins(\Omega, \hi)$ of $\A$ by setting $\I^\A_x(\varrho) = \sqrt{\A(x)} \varrho \sqrt{\A(x)}$ for all $x \in \Omega$ and $\varrho \in \sh$. Thus, each operation $\I^\A_x$ of the Lüders instrumet is defined by a single Kraus operator $\sqrt{\A(x)}$ so it is indecomposable, and $\A^{\I^\A} = \A$. It is a well-known result \cite[Thm. 7.2.]{HayashiBook06} that every instrument $\I \in \ins(\Omega, \hi, \hik)$ with $\A^\I = \A$ can be expressed as $\I_x = \E^{(x)} \circ \I^\A_x$ for some quantum channel $\E^{(x)}\in \ch(\hi, \hik)$ for all $x \in \Omega$. 
\end{example}

\section{Compatibility of quantum devices} \label{sec:3}
\subsection{Definitions of compatibility}
As was mentioned in the introduction, our starting point for this work is what we consider an operational description of parallel compatibility: compatible devices should allow for joint simultaneous implementation meaning that a joint device with a single input should be used to obtain the complete output for each of the original devices at each experimental run. More precisely, the output of the joint device should be a multipartite system composed from the output systems of the original devices such that by ignoring the complete output system of one of them one obtains (with no error) the complete output for the other device. Conceptually parallel compatibility was considered in \cite{HeinosaariMiyaderaZiman16} and mathematically the resulting definitions can also be obtained from the channel compatibility considered in \cite{CarmeliHeinosaariMiyaderaToigo19}.

An important point is that this definition does not reference to any particular type of device but it works with any input-output devices with the only restriction being that the type of the input system is the same (as this must be the type of the input of the joint device as well). This allows us to consider compatibility of otherwise different devices under the same definition. Hence, we can then study compatibility of preparations, measurements, state transformations (i.e. channels) and instruments and all of the possible collections of those devices in the same setting. 

While measurements, channels and instruments have the same type of input 
(a single quantum system), a preparation device prepares a system in a quantum state (i.e. a density matrix) determined by a classical input. Thus, compatibility of preparation devices with these other types of devices is not feasible because of the mismatch of the input types. On the other hand, any preparation devices are always compatible, because classical information can be copied and each of the systems can be then prepared independently. Hence, we will focus on compatibility of measurements, channels and instruments. We will define below when two quantum devices (picked from channels, POVMs and instruments) are compatible and we say that they are incompatible otherwise. 

In quantum theory, measurements with only classical measurement outcomes are described by POVMs. Compatibility of two POVMs defined according to the above principles coincides with the traditional definition of the compatibility of two POVMs, i.e. joint measurability of POVMs, which has a long history and has been studied extensively in the past. Thus, we require the existence of a joint POVM, which outputs a pair of classical outcomes and each of them specifies an outcome of the two POVMs that are jointly measured. In other words, ignoring first (second) outcome provides the outcome of the second (first) POVM from the compatible pair, respectively.

\begin{definition}[POVMs]\label{def:povms}
Two POVMs $\A \in \obs(\Omega,\hi)$ and $\B \in \obs(\Lambda,\hi)$ are compatible, denoted by $\A \comp \B$,  if there exists a joint POVM $\Go \in \obs(\Omega \times \Lambda, \hi)$ such that
\begin{align}
\sum_{x \in \Omega} \Go(x,y) = \B(y) \quad \forall y \in \Lambda, \quad \sum_{y \in \Lambda} \Go(x,y) = \A(x) \quad \forall x \in \Omega.
\end{align}
\end{definition}

Similarly, the definition of compatibility of two quantum channels requires existence of a joint channel that maps input state to a bipartite state on the tensor product of the output systems of the two channels. By ignoring one of the systems on the output one recovers the output of the other channel:

\begin{definition}[Channels]\label{def:channels}
Two channels $\Phi \in \ch(\hi,\hik)$ and $\Psi \in \ch(\hi,\hv)$ are compatible, denoted by $\Phi \comp \Psi$, if there exists a joint channel $\Gamma \in \ch(\hi, \hik \otimes \hv)$ such that
\begin{align}
\ptr{\hik}{\Gamma(\varrho)} = \Psi(\varrho), \quad \ptr{\hv}{\Gamma(\varrho)} = \Phi(\varrho)
\end{align}
for all $\varrho \in \sh$. 
\end{definition}

Compatibility of a channel and a POVM serves as the first example of compatibility of two different types of devices. Since POVMs correspond to measurements with classical measurement outcomes and channels to state transformations, the joint device should be a device that takes a quantum system as an input and outputs both a classical measurement outcome and a transformed post-measurement state. A device of this type is exactly a quantum instrument.

\begin{definition}[POVM and channel]\label{def:povm-ch}
A POVM $\A \in \obs(\Omega,\hi)$ and a channel $\Phi \in \ch(\hi,\hik)$ are compatible, denoted by $\A \comp \Phi$,  if there exists a joint instrument $\G \in \ins(\Omega, \hi, \hik)$ such that
\begin{align}
\Phi^{\G} = \Phi, \quad \A^{\G} = \A.
\end{align}
\end{definition}

We note that trivially for any instrument $\G$ we have that $\A^\G \comp \Phi^\G$. 

If we now follow the definition of parallel compatibility for instruments, the joint device for two compatible instruments should then output two classical measurement outcomes and a state in the tensor product of the output systems of the two instruments. By ignoring a pair of classical-quantum outputs corresponding to one of the instruments one obtains both the classical measurement outcome and the post-measurement state of the other instrument.

\begin{definition}[Instruments]\label{def:instruments}
Two instruments $\I \in \ins(\Omega,\hi,\hik)$ and $\J \in \ins(\Lambda, \hi,\hv)$ are compatible, denoted by $\I \comp \J$, if there exists a joint instrument $\G \in \ins(\Omega \times \Lambda, \hi, \hik \otimes \hv)$ such that
\begin{align}
\sum_{x \in \Omega} \ptr{\hik}{\G_{(x,y)}(\varrho)} = \J_y(\varrho) \quad \forall y \in \Lambda, \quad \sum_{y \in \Lambda} \ptr{\hv}{\G_{(x,y)}(\varrho)} = \I_x(\varrho) \quad \forall x \in \Omega
\end{align}
for all $\varrho \in \sh$. 
\end{definition}

This definition corresponds exactly the definition of \cite{MitraFarkas22} and which can be also recovered from the definition of compatibility presented in \cite{CarmeliHeinosaariMiyaderaToigo19} in the case when the instruments are formulated as channels that map to a mixed quantum-classical system. It is clear that the above definition can be seen as a direct generalization of the definitions of compatible POVMs and channels that were known before: In particular, in Def. \ref{def:instruments} we see that if the considered instruments have just one outcome, i.e. they are channels, then the compatibility criterion reduces to that of Def. \ref{def:channels}. Furthermore, if one considers POVMs as instruments with one-dimensional output spaces, then for two POVMs the compatibility criteria in Def. \ref{def:instruments} reduces to that of Def. \ref{def:povms}. And naturally, when one has a POVM (an instrument with a one-dimensional output space) and a channel (an instrument with only one outcome), then Def. \ref{def:instruments} reduces to Def. \ref{def:povm-ch}.

By applying Def. \ref{def:instruments} in the case of a POVM (an instrument with one-dimensional output space) and some other instrument we see that it generalizes Def. \ref{def:povm-ch} of compatibility between a POVM and a channel. In the case of a POVM and an instrument, the joint device outputs the classical measurement outcome of the POVM as well as the classical and quantum outputs of the instrument; thus, the joint device is also an instrument. We can explicitly state the definition in the following form:

\begin{definition}[POVM and instrument]\label{def:povm-ins}
A POVM $\A \in \obs(\Omega,\hi)$ and an instrument $\J \in \ins(\Lambda, \hi,\hik)$ are compatible, denoted by $\A \comp \J$, if there exists a joint instrument $\G \in \ins(\Omega \times \Lambda, \hi, \hik)$ such that
$\sum_{x \in \Omega} \G_{(x,y)} = \J_y$ for all $y \in \Lambda$ and $\sum_{y \in \Lambda} \A^\G(x,y) = \A(x)$ for all $x \in \Omega$.
\end{definition}

Effectively the above definition means that 1) the original POVM and the induced POVM of the original instrument can be measured jointly and the joint POVM in this case is just the induced POVM of the joint instrument, and 2) the post-measurement state of the original instrument can be recovered from the post-measurement state of the joint instrument. 

In the case when $\J = \Phi$ is a channel, then the above definition of compatibility of a POVM $\A$ and $\Phi$ reduces to Def. \ref{def:povm-ch}. Also, when $\J$ is a POVM, then the above definition reduces to Def. \ref{def:povms}. Thus, together with all the observations that we made above, the compatibility of instruments captures all the other cases of compatibility. 

\begin{remark}\label{remark:A-comp}
In some of the literature for a given POVM $\A$ an instrument $\J$ is said to be $\A$-compatible if $\A^\J = \A$. We note that this is different to saying that $\A$ and $\J$ are compatible because in general we might have $\A^\J \neq \A$ (examples of this type can be constructed from the results of Sec. \ref{sec:6}; see e.g. Prop. \ref{prop:m-a-p-comp} and Prop. \ref{prop:m-a-p-comp1}). In fact in Def. \ref{def:povm-ins} the joint instrument $\G$ not only implements $\J$ as one of its marginals but in fact the induced POVM $\A^\G$ of the joint instrument $\G$ is a joint POVM of $\A$ and $\A^\J$. Thus, if $\A^\J = \A$, then in the notation of Def. \ref{def:povm-ins} one can clearly define the joint instrument as $\G_{(x,y)} = \delta_{x,y} \J_y$ for all $x \in \Omega$ and $y \in \Lambda$. In conclusion, saying that an instrument $\J$ is $\A$-compatible is the same as simply saying that the induced POVM $\A^\J$ of $\J$ is $\A$ or equivalently that $\A$ and $\Phi^\J$ are compatible with one of the joint instruments being $\J$, but in general it is different to saying that $\A$ and $\J$ are compatible.
\end{remark}

One of the most important examples of incompatible instruments/ channels is the fact that two identity channels are incompatible \cite{HeinosaariMiyadera17}; more famously this fact is known as the no-broadcasting theorem \cite{BarnumCavesFuchsJozsaSchumacher96} (which is a generalization of the no-cloning theorem \cite{WoottersZurek82}) since the existence of a joint channel for identity channels is the same as the existence of perfect universal broadcasting device. 

Although we will consider compatibility of specific types of instruments in later sections, we want to start here with a simple class of instruments that are compatible with every other instrument.

\begin{example}\label{ex:trash-and-prepare}
Let $\I \in \ins(\Omega, \hi, \hik)$ be a \emph{trash-and-prepare} (or \emph{trivial}) instrument defined by some probability distribution $(p_x)_{x \in \Omega}$ and some family of states $\{ \xi_x \}_{x \in \Omega} \subset \shik$ such that $\I_x(\varrho) = \tr{\varrho} p_x \xi_x$ for all $\varrho \in \lh$ and $x \in \Omega$. Thus, the action of $\I$ can be simply described as just ignoring or disregarding the input and preparing an output state according to some predetermined probability distribution from a fixed family of states. To see that $\I$ is compatible with any instrument $\J \in \ins(\Lambda, \hi, \hv)$ we can define a joint instrument $\G \in \ins(\Omega \times \Lambda, \hi, \hik \otimes \hv)$ by setting $\G_{(x,y)}(\varrho) = p_x \xi_x \otimes \J_y(\varrho)$ for all $\varrho \in \sh$, $x \in \Omega$ and $y \in \Lambda$.
\end{example}

\subsection{Compatibility of the induced POVMs and channels}
Next, we will explore some of the basic properties of compatible instruments. 
Let $\I$ and $\J$ be two instruments. 
We note that all the above definitions of compatibility are symmetric relations, thus for example $\I \comp \J \Leftrightarrow \J \comp \I$. 
First, we observe that if $\I \comp \J$ by ignoring part of the classical or quantum output of the joint instrument we can prove that $\I$ is compatible with the induced channel $\Phi^\J$ and induced POVM $\A^\J$, respectively.

\begin{proposition}\label{prop:ins-comp-sum}
Let $\I \in \ins(\Omega, \hi, \hik)$ and $\J \in \ins(\Lambda, \hi, \hv)$ be two compatible instruments. Then $\I \comp \Phi^\J$ and $\I \comp \A^\J$
\end{proposition}

\begin{proof}
If $\I \comp \J$, then there exists a joint instrument $\G \in \ins(\Omega \times \Lambda, \hi, \hik \otimes \hv)$ for $\I$ and $\J$. Let us define an instrument $\E \in \ins(\Omega, \hi, \hik \otimes \hv)$ by setting $\E_x = \sum_{y \in \Lambda} \G_{(x,y)}$ for all $x \in  \Lambda$. We then see that
\begin{align*}
\ptr{\hv}{\E_x(\varrho)} &= \sum_{y \in \Lambda} \ptr{\hv}{  \G_{(x,y)}(\varrho)} = \I_x(\varrho)
\end{align*}
for all $x \in \Omega$ and $\varrho\in \sh$, as well as
\begin{align*}
\sum_{x \in \Omega}\ptr{\hik}{\E_x(\varrho)} &= \sum_{y \in \Lambda} \left( \sum_{x \in \Omega} \ptr{\hik}{  \G_{(x,y)}(\varrho)} \right) = \sum_{y \in \Lambda} \J_y(\varrho) = \Phi^\J(\varrho)
\end{align*}
for all $\varrho\in\sh$ so that $\E$ is a joint instrument for $\I$ and $\Phi^\J$ so that indeed $\I \comp \Phi^\J$. 
If $\I \comp \J$, then we define an instrument $\R \in \ins(\Omega \times \Lambda, \hi, \hik)$ by setting
\begin{align*}
\R_{(x,y)}(\varrho) = \ptr{\hv}{\G_{(x,y)}(\varrho)}
\end{align*}
for all $x \in \Omega$, $y \in \Lambda$ and $\varrho \in \sh$. We see that now $\sum_{y \in \Lambda} \R_{(x,y)} = \I_x$ for all $x \in \Omega$ so that $\Phi^\R = \Phi^\I$. Furthermore, we see that $\A^\R = \A^\G$ and since $\A^\G$ is a joint observable for $\A^\I$ and $\A^\J$, so is $\A^\R$ as well so that $\sum_{x \in \Omega} \A^\R(x,y) = \A^\J(y)$ for all $y \in \Lambda$. Thus, $\I \comp \A^\J$. 
\end{proof}

If we consider POVMs as instruments with one dimensional output and channels as instruments with one element outcome set then Proposition \ref{prop:ins-comp-sum} also states that $\A \comp \J \Rightarrow \A \comp \Phi^\J \land \A \comp \A^\J $ and $\Phi \comp \J \Rightarrow \Phi \comp \Phi^\J \land \Phi \comp \A^\J$ for any channel $\Phi$ and any POVM $\A$.  As a direct consequence of Proposition \ref{prop:ins-comp-sum} all the implications in Fig. \ref{fig:1} hold, since we can choose $\A=\A^\I$ and $\Phi=\Phi^\I$.

\begin{figure}
\centering
\includegraphics[scale=0.5]{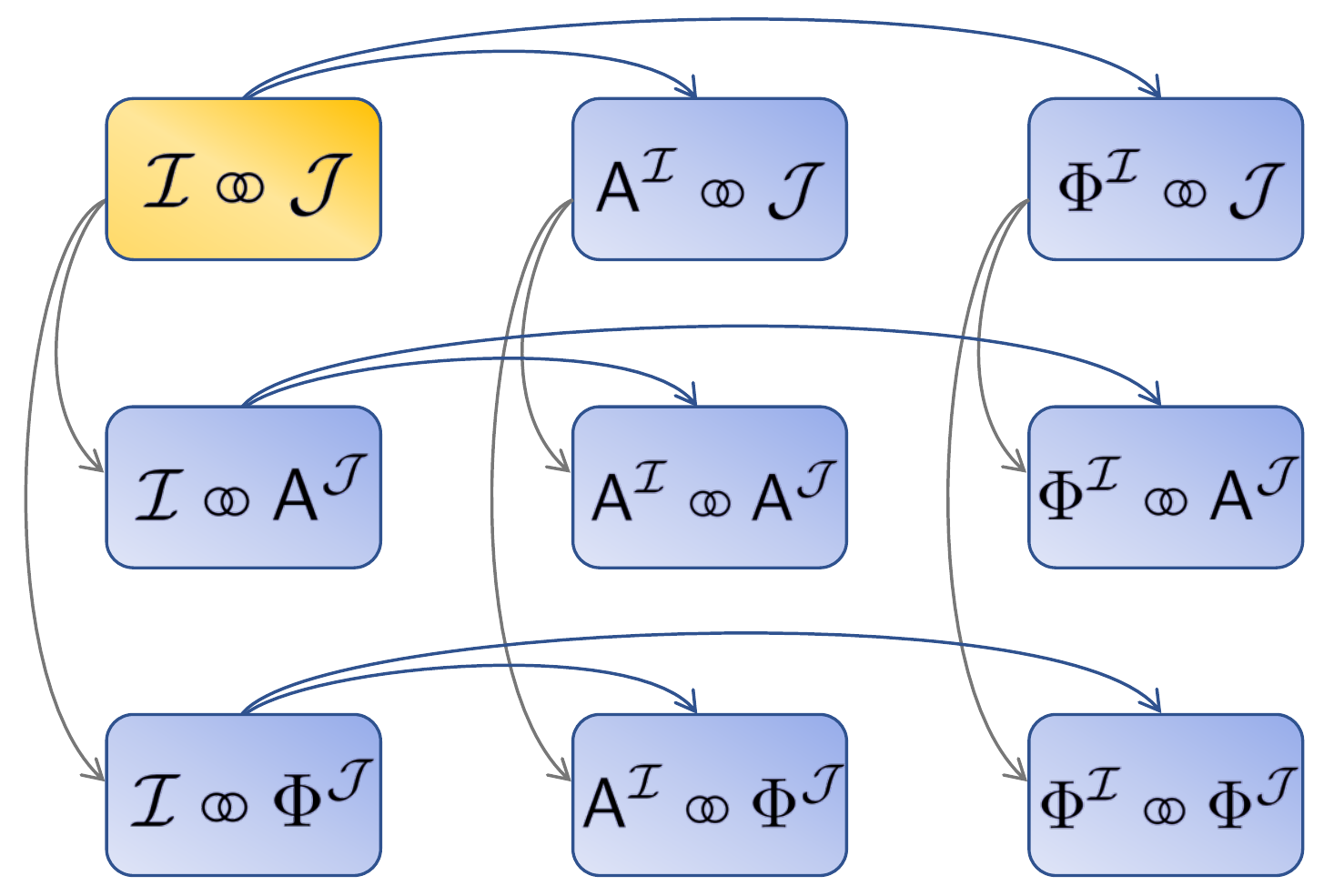}
\caption{\label{fig:1} Implications of the compatibility of instruments $\I$ and $\J$. Compatibility with an instrument automatically implies compatibility with its induced POVM and its induced channel.}
\end{figure}

\begin{corollary}\label{cor:ins-comp-sum}
Let $\I \in \ins(\Omega, \hi, \hik)$ and $\J \in \ins(\Lambda, \hi, \hv)$ be two compatible instruments. Then $X \comp Y$, where $X\in\{\I , \A^\I, \Phi^\I \}$ and $Y\in\{\J , \A^\J, \Phi^\J \}$, i.e. compatibility of instruments implies compatibility also for the induced POVMs or channels on both sides of the relation.
\end{corollary}

We note that some of these implications were also studied in \cite{MitraFarkas22}: in particular, it is shown therein that  $\I \comp \J$ indeed implies that $\A^\I \comp \A^\J$, $\Phi^\I \comp \A^\J$, $\A^\I \comp \Phi^\J$ and $\Phi^\I \comp \Phi^\J$. However, our aim with the above result was to show how the previously unknown notion of compatibility between a POVM and an instrument fits into this picture and also to show all the relevant implications and not only those coming from the compatibility of $\I$ and $\J$.

\subsection{Attempted reductions of the compatibility problem for instruments}
It might be tempting to reduce the compatibility of instruments to the compatibility of its components; namely to the compatibility of both the induced POVMs and channels. Let us show a simple counterexample for this type of reduction.

\begin{example}
Let us consider two qubit instruments $\I \in \ins(\Omega, \hi, \hik)$ and $\J \in \ins(\Lambda, \hi, \hv)$ with input and output Hilbert spaces thus being $\hi= \hik= \hv = \complex^2$: in particular let $\I$ be the single outcome identity instrument $\I = id \in \ch(\complex^2)$, and let $\J$ be a four outcome instrument $\J \in \ins(\{0,1,2,3\},\complex^2)$ defined as $\J_i(\varrho)=\frac{1}{4} \sigma_i \varrho \sigma_i$  for all $y \in \{0,1,2,3\}$, where $\{\sigma_i\}_{i=0}^3$ are the identity and the three Pauli matrices. It is well known that $\Phi^\J(\varrho)=\tr{\varrho}\frac{1}{2} I$ for all $\varrho \in \lh$ so that $\Phi^\J$ is a trash-and-prepare channel. Thus, by Example \ref{ex:trash-and-prepare} we have that $\I = \Phi^\I \comp \Phi^\J$. Since both induced POVMs are trivial, we also have that $\A^\I \comp \A^\J$. On the other hand, the two instruments could not be compatible, because the joint instrument could be used to build perfect universal qubit cloner as follows.  Let $\G \in \ins(\{0,1,2,3\}, \mathbb{C}^2, \mathbb{C}^2\otimes\mathbb{C}^2)$ be the joint instrument. We define quantum channel $\E(\varrho)=\sum_i (I_\hik\otimes\sigma_i) \G_i(\varrho)(I_\hik\otimes\sigma_i)$. 
We then must have for all $\varrho \in \sh$ that
\begin{align*}
\ptr{\hv}{\E(\varrho)}&=\ptr{\hv}{\sum_i \G_i(\varrho)}=\I(\varrho)=\varrho, \nonumber \\
\ptr{\hik}{\E(\varrho)}&=\sum_i \sigma_i (\ptr{\hik}{\G_i(\varrho)})\sigma_i =\sum_i \sigma_i \J_i(\varrho) \sigma_i= \varrho.
\end{align*}
Thus, existence of joint instrument would imply existence of perfect broadcasting channel $\E$, i.e., the joint channel of two identity channels, which is a known contradiction \cite{HeinosaariMiyadera17}.
\end{example}

On the other hand, another tempting reduction might be to reduce the compatibility problem for instruments by removing features that do not affect the compatibility relation. In particular, as was explained in Example \ref{ex:luders}, since every instrument can be realized as a concatenation of the related Lüders instrument and some set of conditional channels, one naturally emerging option would be to reduce the compatibility of the given instruments to the compatibility of the related Lüders instruments. Unfortunately, this is not possible and only one-directional implication works as we show below.

\begin{proposition}\label{prop:comp_Luders}
If $\I \in \ins(\Omega, \hi, \hik)$ and $\J \in \ins(\Lambda, \hi, \hv)$ are two instruments such that their Lüders instruments $\I^{\A^\I}$ and $\I^{\A^\J}$ are compatible, then $\I$ and $\J$ are compatible as well.
\end{proposition}
\begin{proof}
Let us denote by $\{\E^{(x)}\}_{x \in \Omega} \subset \ch(\hi, \hik)$ and $\{\F^{(y)}\}_{y \in \Lambda} \subset \ch(\hi, \hv)$ the sets of conditional channels that postprocess Lüders instruments $\I^{\A^\I}$ and $\I^{\A^\J}$ to $\I$ and $\J$ respectively, i.e., $\I_x=\E^{(x)} \circ \I^{\A^\I}_x$ for all $x \in \Omega$ and $\J_y=\F^{(y)} \circ \I^{\A^\J}_y$ for all $y \in \Lambda$ as per Example \ref{ex:luders}.
Compatibility of $\I^{\A^\I}$ and $\I^{\A^\J}$ implies the existence of a joint instrument $\R \in \ins(\Omega \times \Lambda, \hi, \hi \otimes \hi)$. Based on this, we can define an instrument $\widetilde{\R} \in \ins(\Omega \times \Lambda, \hi, \hik \otimes \hv)$ by setting $\widetilde{\R}_{(x,y)}= \left( \E^{(x)} \otimes \F^{(y)} \right) \circ \R_{(x,y)}$ for all $x \in \Omega$ and $y \in \Lambda$. We see that
\begin{align*}
\sum_y \ptr{\hv}{\widetilde{\R}_{(x,y)}(\varrho)}&=\sum_y \ptr{\hv}{\left( \E^{(x)} \otimes \F^{(y)} \right)\left(\R_{(x,y)}(\varrho)\right)} \\
 &= \E^{(x)} \circ \ptr{\hi}{\sum_y \R_{(x,y)}(\varrho)}=\E^{(x)} \circ \I^{\A^\I}_x(\varrho) \\
 &= \I_x(\varrho)
\end{align*}
for all $\varrho \in \sh$ and for all $x \in \Omega$. Analogically we get that $\sum_x \ptr{\hik}{\widetilde{\R}_{(x,y)}}=\J_y$ for all $y \in \Lambda$ so that $\widetilde{\R}$ is a joint instrument for $\I$ and $\J$.
\end{proof}

On the other hand, to see that the converse does not hold, let us consider a pair of single outcome instruments (i.e. channels) defined as $\I(\varrho)=\varrho$ and $\J(\varrho)=\tr{\varrho}\xi$ for all $\varrho$. These instruments are compatible, but their Lüders instruments (identity channel in both cases) are incompatible since cloning is again impossible.

\section{Compatibility and non-disturbance}\label{sec:4}

For POVMs compatibility is also closely related to the concept of non-disturbance. Traditionally the non-disturbance for POVMs (see e.g. \cite{HeinosaariWolf10}) is defined as follows: a POVM $\A \in \obs(\Omega, \hi)$ does not disturb a POVM $\B \in \obs(\Omega, \hi)$ if there exists an instrument $\I \in \ins(\Omega, \hi)$ with $\A^\I = \A$ such that
\begin{align*}
\tr{\B(y) \Phi^\I(\varrho)} = \tr{\B(y) \varrho}
\end{align*}
for all $y \in \Lambda$ and $\varrho \in \sh$. The previous condition does indeed capture the idea that $\A$ can be measured by some instrument $\I$ such that the measurement statistics of $\B$ is not disturbed by the transformation induced by the instrument $\I$. Equivalently we can phrase the non-disturbance condition as a condition on the adjoint map $(\Phi^\I)^*$; we simply require that $\left( \Phi^\I \right)^*(\B(y)) = \B(y)$ for all $y \in \Lambda$.

However, strictly speaking the traditional non-disturbance condition for POVMs does not involve the POVM that is supposed to be non-disturbing but rather \emph{a way of measuring} that POVM via some instrument. This is because for non-disturbance one needs to be able to do sequential operations on the post-measurement state. Furthermore, as this possibility is inherent in the description of an instrument, a natural generalization from non-disturbance of POVMs is to consider the non-disturbance of instruments. Since instruments have another output in addition to the classical measurement outcome, namely the output state, then also this output must remain undisturbed. Hence, we make the following definition:

\begin{definition}[Non-disturbance]
\label{def:nondisturb}
An instrument $\I \in \ins(\Omega, \hi)$ is said \emph{not to disturb an instrument $\J \in \ins(\Lambda, \hi, \hv)$}, denoted by $\I \nd \J$, if $\J_y \circ \Phi^\I = \J_y$ for all $y \in \Lambda$. We also say that a POVM $\A \in \obs(\Omega, \hi)$ \emph{does not disturb an instrument} $\J \in \ins(\Lambda, \hi, \hv)$, denoted by $\A \nd \J$, if there exists an instruments $\I \in \ins(\Omega, \hi)$ with $\A^\I= \A$ such that $\I \nd \J$.
\end{definition}

It is straightforward to verify that if $\J$ is a POVM, i.e., an instrument with one-dimensional output space, then the latter part of the previous definition reduces to the traditional non-disturbance condition that was described above. However, we note that we have to inevitably use separate definitions for POVM and instrument non-disturbance and we cannot simply consider non-disturbing POVM as an instrument with one-dimensional output space. This is because a non-disturbing device must have its input and output systems equal to the input system of the undisturbed device so that the sequential implementation is possible. Especially, as the input systems of POVMs, channels and instruments are non-trivial quantum systems, non-disturbance of POVMs must be phrased in terms of the instruments implementing them. On the other hand, for channels we use the above definition to say that a channel does not disturb a POVM/channel/instrument by considering the channel as a one-outcome instrument. Consequently, we see that an instrument $\I$ does not disturb a POVM/channel/instrument if and only if the channel $\Phi^\I$ does not disturb it.

Based on the definition \ref{def:nondisturb} it is easy to see the following connections between non-disturbance of different devices.

\begin{proposition}\label{prop:nd-implications}
Let $\I \in \ins(\Omega, \hi)$ and $\J \in \ins(\Lambda, \hi)$ be two instruments. Then the implications in Fig. \ref{fig:2} hold.
\end{proposition}

\begin{figure}
\centering
\includegraphics[scale=0.5]{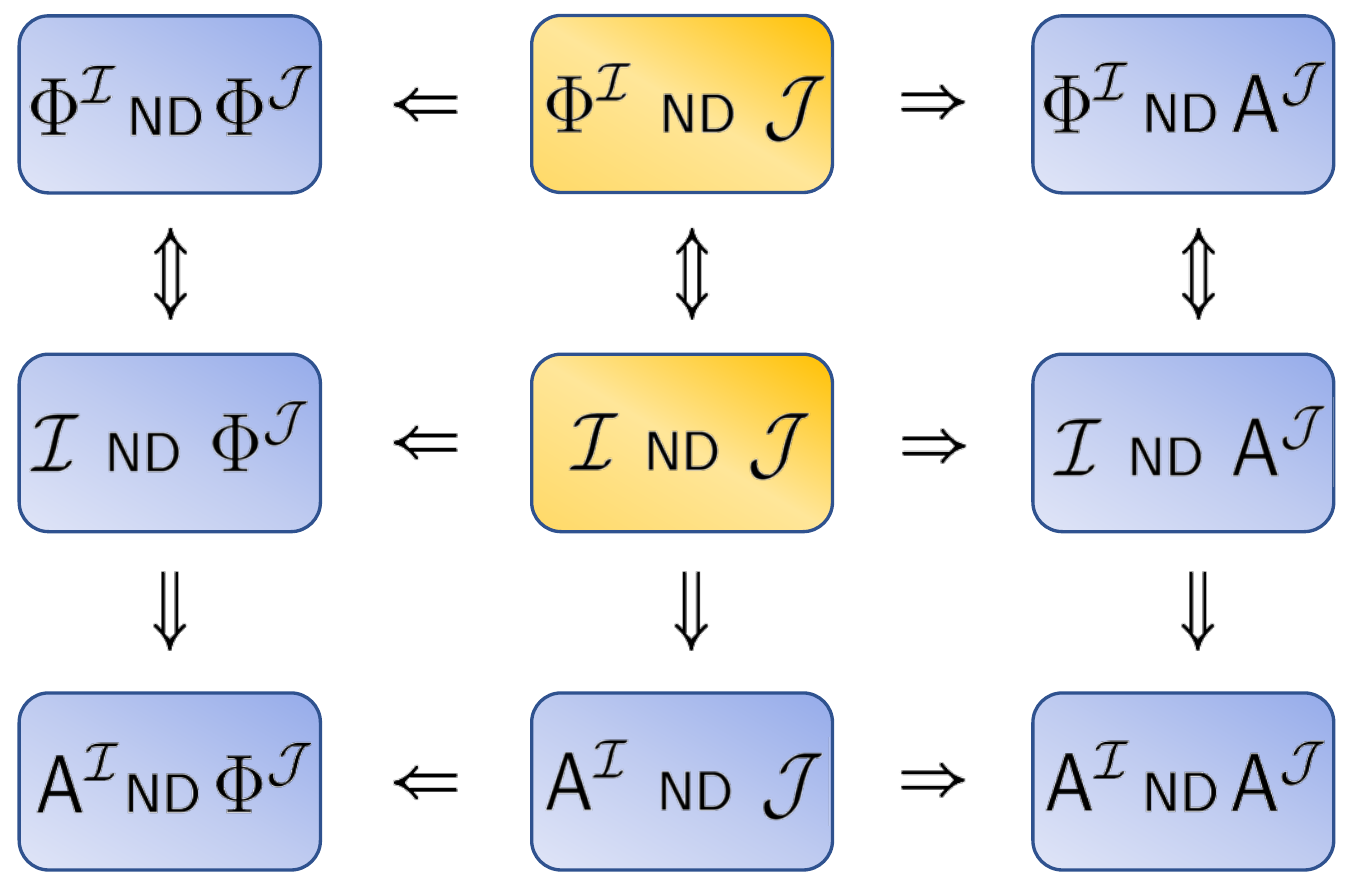}
\caption{\label{fig:2} Implications of the 
non-disturbance of instruments $\I$ and $\J$. Non-disturbance of an instrument automatically implies non-disturbance of  its induced POVM and its induced channel.}
\end{figure}

It is known that there is a clear connection between non-disturbance and compatibility in the case of POVMs (see e.g. \cite{HeinosaariWolf10}). In particular, if a POVM is not disturbed by a measurement of another POVM, then they can be measured jointly sequentially so that they are compatible. Interestingly, if one of the POVMs is sharp, then their compatibility is also sufficient for the sharp POVM not to disturb the other one. We can generalize this known connections to the situation when a POVM does not disturb an instrument.

\begin{proposition}\label{prop:nd-comp}
If a POVM $\A$ does not disturb an instrument $\J$, then $\A$ and $\J$ are compatible. Furthermore, for sharp $\A$ also the converse holds.
\end{proposition}
\begin{proof}
Let $\A \in \obs(\Omega, \hi)$ and $\J \in \ins(\Lambda, \hi, \hv)$. If $\A \nd \J$, then there exists an instrument $\I \in \ins(\Omega, \hi)$ such that $\A^\I = \A$ and $\J_y \circ \Phi^\I = \J_y$ for all $y \in \Lambda$. We can define a joint instrument $\G \in \ins(\Omega \times \Lambda, \hi, \hv)$ for $\A$ and $\J$ by setting $\G_{(x,y)} = \J_y \circ \I_x$ for all $x \in \Omega$ and $y \in \Lambda$.

Let now $\A$ be sharp. If $\A \comp \J$, then there exists an instrument $\G \in \ins(\Omega \times \Lambda, \hi,\hv)$ such that $\sum_{y \in \Lambda} \A^\G(x,y) = \A(x)$ for all $x \in \Omega$ and $\sum_{x \in \Omega} \G_{(x,y)} = \J_y$ for all $y \in \Lambda$. Thus, $\Go:=\A^\G$ is a joint POVM for $\A$ and $\B:= \A^\J$, and because $\A$ is sharp then it follows from \cite{HeinosaariReitznerStano08} that they commute and they have a unique joint POVM of the form $\Go(x,y) = \A(x) \B(y)$ for all $x \in \Omega$ and $y \in \Lambda$. Furthermore, as is explained in Example \ref{ex:luders}, since $\A^\G = \Go$, there exist a set of channels $\{\E^{(x,y)}\}_{x \in \Omega, y \in \Lambda} \subset \ch(\hi, \hv)$ such that $\G_{(x,y)} = \E^{(x,y)} \circ \I^{\Go}_{(x,y)}$ for all $x \in \Omega$ and $y \in \Lambda$. We note that from the form of $\Go$ and the commutativity of $\A$ and $\B$ it follows that $\I^{\Go}_{(x,y)} = \I^\A_x \circ \I^\B_y = \I^\B_y \circ \I^\A_x$ for all  $x \in \Omega$ and $y \in \Lambda$. Furthermore, from the sharpness of $\A$ it also follows that $\I^\A_x \circ \I^\A_{x'} = \delta_{xx'} \I^\A_x$ for all $x,x' \in \Omega$. Thus, we have that
\begin{align*}
\J_y \circ \Phi^{\I^\A} &= \sum_{x, x' \in \Omega} \G_{(x,y)} \circ \I^\A_{x'} =  \sum_{x, x' \in \Omega} \E^{(x,y)} \circ \I^{\Go}_{(x,y)}  \circ \I^\A_{x'} \\
&= \sum_{x, x' \in \Omega} \E^{(x,y)} \circ  \I^\B_y \circ \I^\A_x  \circ \I^\A_{x'} \\
&= \sum_{x \in \Omega} \E^{(x,y)} \circ  \I^\B_y \circ \I^\A_x  \\
&= \sum_{x \in \Omega} \G_{(x,y)} = \J_y
\end{align*}
for all $y \in \Lambda$ so that the Lüders instrument of $\A$ does not disturb the instrument $\J$.
\end{proof}

By taking $\J$ to be some POVM $\B$, the known result $\A \nd \B \Rightarrow \A \comp \B$ (see e.g. \cite{HeinosaariWolf10}) can be recovered, and similarly by taking $\J$ to be some channel $\Phi$, we also recover the implication $\A \nd \Phi \Rightarrow \A \comp \Phi$. Naturally for sharp $\A$ their converses also hold.

Finally, we can combine some of the results from Prop. \ref{prop:ins-comp-sum}, Prop. \ref{prop:nd-implications} and Prop. \ref{prop:nd-comp}.

\begin{corollary}
Let $\I \in \ins(\Omega, \hi)$ and $\J \in \ins(\Lambda, \hi, \hv)$ be two instruments. Then the following implications hold:
\begin{equation*}
    \I \nd \J \ \Rightarrow \ \A^\I \nd \J \ \Rightarrow \ \A^\I \comp \J \ \Rightarrow \A^\I \comp \Phi^\J \ \mathrm{and} \ \A^\I \comp \A^\J
\end{equation*}
\end{corollary}

Again by choosing the instrument $\J$ to be a POVM or a channel we can recover the respective special cases.

\section{Postprocessing and compatibility} \label{sec:5}

For POVMs it is known that compatibility can be equivalently stated in terms of \emph{postprocessing} (see e.g. \cite{AliCarmeliHeinosaariToigo09}): POVMs $\A \in \obs(\Omega, \hi)$ and $\B \in \obs(\Lambda,\hi)$ are compatible if and only if both of them can be postprocessed from a single observable $\C \in \obs(\Gamma, \hi)$, i.e., there exist stochastic matrices $\nu^\A:=(\nu^\A_{zx})_{z \in \Gamma, x \in \Omega}$ and $\nu^\B:=(\nu^\B_{zy})_{z \in \Gamma, y \in \Lambda}$ (so that $\nu^\A_{zx}\geq 0$, $\nu^\B_{zy}\geq 0$ and $\sum_x \nu^\A_{zx} = \sum_y \nu^\B_{zy} =1$ for all $x \in \Omega$, $y \in \Lambda$ and $z \in \Gamma$) such that $\A(x) = \sum_{z \in \Gamma} \nu^\A_{zx} \C_z$ for all $x \in \Omega$ and  $\B(y) = \sum_{z \in \Gamma} \nu^\B_{zy} \C_z$ for all $y \in \Lambda$. In this case we denote $\A \preceq \C$ and $\B \preceq \C$. In this section we investigate the relation between compatibility and postprocessing for quantum instruments.

\subsection{Postprocessing of instruments}

As a generalization of POVM and channel postprocessing, we recall the following definition of postprocessing of instruments from \cite{LeppajarviSedlak21}:

\begin{definition}
An instrument $\I \in \ins(\Omega, \hi, \hik)$ is a postprocessing of an instrument $\J \in \ins(\Lambda, \hi, \hv)$, denoted by $\I \preceq \J$, if there exist instruments $\{\R^{(y)} \}_{y \in \Lambda} \subset \ins(\Omega, \hv, \hik)$ such that $\I_x = \sum_{y \in \Lambda} \R^{(y)}_x \circ \J_y$ for all $x \in \Omega$.
\end{definition}

By considering POVMs as instruments with one-dimensional output space, it can be shown (see \cite[Prop. 13]{LeppajarviThesis} for an analogous proof) that the previous definition reduces to the definition of postprocessing of the POVMs which was described above. For channels $\Phi \in \ch(\hi, \hik)$ and $\Psi \in \ch(\hi, \hv)$, we can express the known definition as follows:  $\Phi \preceq \Psi$ if there exists a channel $\Phi' \in \ch(\hv, \hik)$ such that $\Phi = \Phi' \circ \Psi$. This has been previously studied in \cite{HeinosaariMiyadera13,Jencova21,BenyOreshkov11} and it can also be interpreted as an instrument postprocessing in which the postprocessing instrument is just a single channel $\Phi'$. 

We also want to note that one can consider postprocessings between channels and instruments: an instrument $\I \in \ins(\Omega, \hi, \hik)$ is a postprocessing of a channel $\Psi \in \ch(\hi, \hv)$, i.e. $\I \preceq \Psi$ if there exists an instrument $\Q \in \ins(\Omega, \hv, \hik)$ such that $\I_x = \Q_x \circ \Psi$ for all $x \in \Omega$, and on the other hand,  $\Psi \preceq \I$ if there exists channels $\{\Phi^{(x)}\}_{x \in \Omega} \subset \ch(\hik, \hv)$ such that $\Psi = \sum_{x \in \Omega} \Phi^{(x)} \circ \I_x$.

Next, we observe that the compatibility condition for instruments as presented in Def. \ref{def:instruments} implies also a postprocessing relation between the joint instrument and the original instruments, i.e. $\I \comp \J \Rightarrow \exists \G: \I \preceq \G \land \J \preceq \G$. This holds, because making a partial trace and ignoring one of the outcomes/outputs of the joint instrument is a valid operation in the instrument postprocessing sense. Thus, being postprocessings of a single instrument constitutes a necessary condition for instrument compatibility. It clearly is not a sufficient condition (unlike in the case of POVMs), because for instruments there exists (see \cite{HeinosaariMiyadera13} and \cite{LeppajarviSedlak21}) a postprocessing greatest element (e.g. the identity channel). From this element every other channel and instrument can be postprocessed. Thus, if the existence of a common upper bound would imply compatibility of instruments, then all channels and instruments should be compatible, which is not true as the example of two identity channels shows. 

Although existence of a common upper bound in the postprocessing order is a necessary condition for instrument compatibility, 
unfortunately, it is not very practical, because narrowing down the search for joint instrument to all common upper bounds of the two original instruments might be a very complex task. Nevertheless, we can still find connections between the postprocessing and the compatibility relations. 

\begin{proposition}\label{prop:ins-pp-comp}
Let $\I \in \ins(\Omega, \hi,\hik)$ and $\J \in \ins(\Lambda, \hi, \hv)$ be two instruments. If $\I \preceq \J$, then $\I \comp \A^\J$, and in particular $\Phi^\I \comp \A^\J$ and $\A^\I \comp \A^\J$.
\end{proposition}
\begin{proof}
If $\I \preceq \J$, then there exist instruments $\{\R^{(y)} \}_{y \in \Lambda} \subset \ins(\Omega, \hv, \hik)$ such that $\I_x = \sum_{y \in \Lambda} \R^{(y)}_x \circ \J_y$ for all $x \in \Omega$. If we now define an instrument $\G \in \ins(\Omega \times \Lambda, \hi, \hik)$ by setting $\G_{(x,y)} = \R^{(y)}_x \circ \J_y$ for all $x \in \Omega$ and $y \in \Lambda$, we can confirm that $\sum_{y \in \Lambda} \G_{(x,y)} = \I_x$ for all $x \in \Lambda$ and that $\sum_{x \in \Lambda} \A^\G(x,y) = \A^\J(y)$ for all $y \in \Lambda$. Thus, $\G$ is a joint instrument for $\I$ and $\A^\J$ so that $\I \comp \A^\J$. It follows from Prop. \ref{prop:ins-comp-sum} that then also $\A^\I \comp \A^\J$ and $\Phi^\I \comp \A^\J$.
\end{proof}

The above result resembles a similar result for POVMs, which states that if one POVM can be postprocessed from another, then they must be compatible as they can be both postprocessed from a single POVM. However, the main difference between these results is that for instruments the postprocessing relation does not concern having multipartite quantum output, since it only describes sequential operations while the compatibility relation for instruments does require multipartite quantum output. For this reason the postprocessing relation in the above proposition does not guarantee compatibility of the considered instruments, but only compatibility of an induced POVM and the other instrument.

Luckily, it turns out that there is a way to get a necessary and sufficient instrument postprocessing condition for compatibility of instruments, although it considers postprocessings of so-called complementary instruments that we consider next.

\subsection{Complementary instrument}\label{sec:complementary-ins}
We recall that each Stinespring dilation $(\hi_A, W)$ of a channel $\Phi \in \ch(\hi, \hik)$ defines a \emph{complementary channel} (or \emph{conjugate channel}) $\Phi^C \in \ch(\hi, \hi_A)$ by tracing out the output Hilbert space $\hik$ instead of the ancillary Hilbert space $\hi_A$ in Eq. \eqref{eq:stinespring-def} so that 
\begin{equation*}
    \Phi^C(\varrho) = \ptr{\hik}{W \varrho W^*}
\end{equation*}
for all $\varrho \in \sh$. Similarly, we want to define an analogous concept for instruments and for this we first need to consider dilations of instruments. 

\begin{definition}[Dilation of an instrument]\label{def:MinDilInst}
A \emph{dilation of an instrument} $\I \in \ins(\Omega, \hi, \hik)$ is a triple $(\hi_A, W, \Eo)$ consisting of
a Hilbert space  $\hi_A$,  an isometry $W : \hi \rightarrow \hi_A \otimes \hik$ and a POVM $\Eo\in \obs(\Omega,\hi_A)$
such that
\begin{align}\label{eq:inst-dil-def}
\I_x(\varrho)=\ptr{\hi_A}{W\varrho W^* \left(\Eo(x) \otimes I_{\hik} \right)}
\end{align}
 for all $\varrho \in \sh$ and $x \in \Omega$. Then Eq. \eqref{eq:inst-dil-def} implies that $(\hi_A, W)$ is a Stinespring dilation of the channel $\Phi^\I$, and we say that the dilation $(\hi_A, W, \Eo)$ of $\I$ is minimal if the Stinespring dilation $(\hi_A,W)$ of $\Phi^\I$ is minimal.
\end{definition}

\begin{remark}\label{remark:min-dil-unique-povm}
Although the definition of a minimal dilation of an instrument only considers the minimality of the Stinespring dilation of the induced channel, we note that it still has effects on the POVM in the dilation. For example, one sees that for a given minimal dilation $(\hi_A, W, \Eo)$ of an instrument $\I \in \ins(\Omega, \hi, \hik)$ the POVM $\Eo \in \obs(\Omega, \hi_A)$ is unique: Let $(\hi_A, W, \Fo)$ be a dilation of $\I$ for some POVM $\Fo \in \obs(\Omega,\hi_A)$. Clearly, since $(\hi_A, W, \Eo)$ is minimal then so is $(\hi_A, W, \Fo)$ as well. Since they are dilations for the same instrument $\I$, by using Eq. \eqref{eq:inst-dil-def} for both of the dilations it follows that $W^*(\Eo(x) \otimes X) W = W^*(\Fo(x) \otimes X) W$ for all $X \in \lk$ and $x \in \Omega$. It is straightforward to see from the minimality condition it follows that we must have $\Eo = \Fo$. This can also be seen as an immediate corollary of the Radon-Nikodym theorem for quantum operations \cite{Raginsky03}. In particular, if $(\hi'_A, W', \Eo')$ is any other dilation of $\I$, then we get that it is linked to the minimal dilation $(\hi_A, W, \Eo)$ by an isometry $V: \hi_A \to \hi'_A$ such that $W'=(V \otimes I_\hik)W$, and furthermore, $\Eo(x) = V^* \Eo'(x) V$ for all $x \in \Omega$.
\end{remark}

In the literature, a dilation is sometimes also called a \emph{measurement model} (see. e.g. \cite{HeinosaariZimanBook}). We note that dilation can also be defined differently, for example by requiring that the POVM $\Eo$ is sharp. In particular, the usual Stinespring type representation for completely positive instruments requires that $\Eo$ is sharp. However, for our purposes Definition \ref{def:MinDilInst} is  more convenient. It was shown in \cite{Ozawa84} that each quantum instrument indeed has a dilation, and even that the POVM on the ancillary Hilbert space can be chosen to be sharp (in this case the dilation is not necessarily minimal). For finite-dimensional Hilbert spaces the minimal dilation of an instrument can always be explicitly constructed (see e.g. \cite{ChiribellaD'ArianoPerinotti09}).

\begin{definition}[Complementary instrument]\label{def:ComplInst}\leavevmode
For an instrument $\I \in \ins(\Omega, \hi, \hik)$ with a dilation $(\hi_A, W, \Eo)$ we define \emph{complementary instrument} $\I^C \in \ins(\Omega, \hi, \hi_A)$ relative to this dilation via the formula
\begin{align}
\I^C_x(\varrho)=\ptr{\hik}{\left(\sqrt{\Eo(x)}\otimes I_{\hik} \right) W\varrho W^*  \left(\sqrt{\Eo(x)}\otimes I_{\hik}\right) }
\end{align}
for all $\varrho \in \sh$ and $x \in \Omega$.
\end{definition}

We note that complementary instrument can also be interpreted as the complementary channel $\left(\Phi^\I \right)^C$ of $\Phi^\I$ followed by the Lüders instrument $\I^\Eo$ of the POVM $\Eo$ from the dilation of $\I$. Indeed, it is easy to confirm that
\begin{equation}\label{eq:compl-ins}
    \I^C_x = \I^\Eo_x \circ \left(\Phi^\I \right)^C
\end{equation}
for all $x \in \Omega$ and any dilation of $\I$.

\begin{example}\label{ex:compl-ins-povm}
Let $\A \in \obs(\Omega, \hi)$ be a POVM. We can consider  $\A$ as an instrument $\mathcal{A} \in \ins(\Omega, \hi, \complex)$ with a one-dimensional output space $\complex$ such that $\mathcal{A}_x(\varrho) = \tr{\A(x) \varrho}$ for all $\varrho \in \sh$ and $x \in \Omega$. We can define a trivial dilation $(\hi_A, W, \Eo)$ for $\mathcal{A}$ by taking $\hi_A = \hi$, $W\varphi = \varphi \otimes I_\complex$ for all $\varphi \in \hi$ and $\Eo =\A$. With these choices we have that $\mathcal{A}_x(\varrho) = \ptr{\hi_A}{W \varrho W^*(\Eo(x) \otimes I_\complex)}$ for all  $\varrho \in \sh$ and $x \in \Omega$. On the other hand, the complementary instrument $\mathcal{A}^C \in \ins(\Omega, \hi)$ related to this dilation is
\begin{align*}
    \mathcal{A}_x^C(\varrho) &= \ptr{\complex}{\left(\sqrt{\Eo(x)} \otimes I_\complex \right) W \varrho W^* \left(\sqrt{\Eo(x)} \otimes I_\complex \right)} \\
    &= \ptr{\complex}{\sqrt{\A(x)}\varrho\sqrt{\A(x)} \otimes I_\complex} = \sqrt{\A(x)}\varrho\sqrt{\A(x)} = \I^\A_x(\varrho)
\end{align*}
 for all  $\varrho \in \sh$ and $x \in \Omega$. Thus, one explicit form of the complementary instrument of a POVM is the Lüders instrument of this POVM.
\end{example}

\subsection{Postprocessing of complementary instruments}

For two quantum channels it was shown in \cite{HeinosaariMiyadera17} that they are compatible if and only if one of them can be postprocessed from any complementary (i.e. conjugate) channel of the other one. More formally, \emph{if $\Phi \in \ch(\hi,\hik)$ and $\Psi \in \ch(\hi, \hv)$ are two channels, then they are compatible, $\Phi \comp \Psi$, if and only if $\Phi \preceq \Psi^C$ (or equivalently $\Psi \preceq \Phi^C$) for any complementary channel $\Psi^C$ ($\Phi^C)$ related to any Stinespring dilation of the channel $\Psi$ ($\Phi$)}. The reason that the relation holds for any complementary channels follows from the fact that any two complementary channels related to two different Stinespring dilations of the same channel are postprocessing equivalent \cite{HeinosaariMiyadera17}.

In this section we will generalize the above compatibility condition in the case of instruments. First, we show that as in the case of channels, also for instruments the different complementary instruments of the same instrument are postprocessing equivalent.

\begin{proposition}\label{prop:ins-dil-pp-equiv}
Let $\I \in \ins(\Omega, \hi, \hik)$ be an instrument. Then all complementary instruments related to different dilations of $\I$ are postprocessing equivalent.
\end{proposition}
\begin{proof}
Let $(\hi_A, W, \Eo)$ be any minimal dilation of $\I$ and let us denote the complementary instrument related to this dilation by $\I^C$ so that $\I^C \in \ins(\Omega, \hi, \hi_A)$. We show that a complementary instrument $\I^{C'}\in \ins(\Omega, \hi, \hi'_A)$  related to any other dilation $(\hi'_A, W', \Eo')$ (which is not necessarily minimal) is postprocessing equivalent to the complementary instrument $\I^{C}$. Since the postprocessing relation is transitive this then shows that two complementary instruments related to any two dilations of $\I$ must also be postprocessing equivalent.

Since $(\hi_A, W,\Eo)$ is a minimal dilation of $\I$, i.e., $(\hi_A, W)$ is a minimal Stinespring dilation of $\Phi^\I$, then there exists an isometry $Y: \hi_A \to \hi'_A$ such that $W'=(Y \otimes I)W$. For the complementary channels $(\Phi^\I)^{C}$ and $(\Phi^\I)^{C'}$ it is easy to check that then
\begin{equation}\label{eq:compl-ch-1}
    \left(\Phi^\I\right)^{C'} = \Y \circ \left(\Phi^\I\right)^{C},
\end{equation}
where we have introduced the concatenation channel $\Y \in \ch(\hi_A,\hi'_A)$ defined as $\Y(\varrho) = Y \varrho Y^*$ for all $\varrho \in \mathcal{L}(\hi_A)$. Since $Y$ is an isometry the above relation can be reversed and we also have that
\begin{equation}\label{eq:compl-ch-2}
    \left(\Phi^\I\right)^{C} = \Y^* \circ \left(\Phi^\I\right)^{C'},
\end{equation}
where now $\Y^*$ is the adjoint map of $\Y$, i.e., $\Y^*(\varrho) = Y^* \varrho Y$ for all $\varrho \in \mathcal{L}(\hi'_A)$. Furthermore, as in Remark \ref{remark:min-dil-unique-povm}, in this case for the POVMs it follows from the minimality of $(\hi_A, W, \Eo)$ that $\Eo = \Y^* \circ \Eo'$, i.e., $\Eo(x) = Y^* \Eo'(x) Y$ for all $x \in \Omega$.

Let us now consired the instrument $\J:=\I^{\Eo'} \circ \Y \in \ins(\Omega, \hi_A, \hi'_A)$. We see that now $\A^\J = \Y^* \circ \Eo' = \Eo$. Thus, in particular $\A^\J$ is postprocessing equivalent with $\Eo$ so that by the results of \cite[Prop. 9]{LeppajarviSedlak21} it follows that the indecomposable instruments $\J$ and $\I^{\Eo}$ are postprocessing equivalent. Thus, there exist two sets of instruments $\{\R^{(x)}\}_{x \in \Omega} \subset \ins(\Omega, \hi_A, \hi'_A)$ and $\{\R'^{(x')}\}_{x' \in \Omega} \subset \ins(\Omega, \hi'_A, \hi_A)$ such that 
\begin{align}
    \J_{x'} &=  \sum_{x \in \Omega} \R^{(x)}_{x'} \circ \I^{\Eo}_x, \quad \forall x' \in \Omega, \label{eq:E-luders-to-J} \\
    \I^{\Eo}_{x} &= \sum_{x' \in \Omega} \R'^{(x')}_{x} \circ \J_{x'}    \quad \forall x \in \Omega. \label{eq:J-to-E-luders}
\end{align}

Now using Eq. \eqref{eq:compl-ins} for both dilations along with Eqs. \eqref{eq:compl-ch-1} and \eqref{eq:J-to-E-luders} we see that 
\begin{align*}
    \sum_{x' \in \Omega} \R'^{(x')}_x \circ \I^{C'}_{x'} &= \sum_{x' \in \Omega} \R'^{(x')}_x \circ \I^{\Eo'}_{x'} \circ \left(\Phi^\I \right)^{C'} \\
    &= \sum_{x' \in \Omega} \R'^{(x')}_x \circ \I^{\Eo'}_{x'} \circ  \Y \circ \left(\Phi^\I\right)^{C} \\
    &= \sum_{x' \in \Omega} \R'^{(x')}_x \circ \J_{x'} \circ \left(\Phi^\I\right)^{C} \\
    &= \I^{\Eo}_{x} \circ \left(\Phi^\I\right)^{C} = \I^{C}_x
\end{align*}
for all $x \in \Omega$. Hence, $\I^{C} \preceq \I^{C'}$.

Similarly with Eqs. \eqref{eq:compl-ch-2} and \eqref{eq:E-luders-to-J} we see that
\begin{align*}
    \sum_{x \in \Omega} \R^{(x)}_{x'} \circ \I^{C}_{x} &= \sum_{x \in \Omega} \R^{(x)}_{x'} \circ \I^{\Eo}_{x} \circ \left(\Phi^\I \right)^{C} \\
    &= \sum_{x \in \Omega} \R^{(x)}_{x'} \circ \I^{\Eo}_{x} \circ  \Y^* \circ \left(\Phi^\I\right)^{C'} \\
    &=  \J_{x'} \circ  \Y^* \circ \left(\Phi^\I\right)^{C'} \\
    &= \I^{\Eo'}_{x'} \circ \Y \circ  \Y^* \circ \left(\Phi^\I\right)^{C'} \\
    &= \I^{\Eo'}_{x'} \circ \left(\Phi^\I\right)^{C'} = \I^{C'}_{x'}
\end{align*}
for all $x' \in \Omega$, where moving on to the last line we have combined Eqs. \eqref{eq:compl-ch-1} and \eqref{eq:compl-ch-2}. Hence, also $\I^{C'} \preceq \I^{C}$ so that in fact $\I^{C'}$ is postprocessing equivalent with $\I^{C}$.
\end{proof}

Now we are ready to state and prove that the problem of verifying the compatibility of a pair of instruments can be equivalently recast as checking the postprocessing relation among one of the instruments and the complementary instrument of the other.

\begin{proposition}\label{prop:IncompAndPostproc}
Let $\I \in \ins(\Omega, \hi, \hik)$ and $\J \in \ins(\Lambda, \hi, \hv)$ be instruments and let $\I^C \in \ins(\Omega, \hi, \hi_A)$ be a complementary instrument related to any dilation $(\hi_A, W, \Eo)$. Then the following are equivalent:
\begin{itemize}
    \item[\textit{i)}]  $\I \comp \J$,
    \item[\textit{ii)}] $\J \preceq \I^C$.
\end{itemize}
\end{proposition}
\begin{proof}
Let us start by considering any dilation $(\hi_A, W, \Eo)$ of the instrument $\I$, which gives us a complementary instrument $\I^C$. If $\J \preceq \I^C$ then there exist postprocessing instruments $\R^{(x)} \in \ins(\Lambda,\hi_A, \hv)$ such that $\J_y(\varrho)=\sum_{x \in \Omega} \R^{(x)}_y(\I^C_x(\varrho))$ for all $\varrho\in\sh$ and $y \in \Lambda$. We use it to define the joint instrument $\G \in \ins(\Omega \times \Lambda, \hi, \hv \otimes \hik)$ as follows:
\begin{align}
\G_{(x,y)}(\varrho) := ( \R^{(x)}_y \otimes id_{\hik})\left(\left(\sqrt{\Eo(x)} \otimes I_{\hik} \right)  W \varrho \;W^* \;\left(\sqrt{\Eo(x)} \otimes I_{\hik} \right)\right)
\end{align}
for all $\varrho \in \sh$, $x \in \Omega$ and $y \in \Lambda$. The above can be seen as a composition of three instruments (an isometry + Lüders instrument + one chosen postprocessing instrument), thus $\G$ is a valid instrument. Let us now investigate the marginals of $\G$. In particular, for all $\varrho \in \sh$ we have that
\begin{align}
\ptr{\hv}{\sum_y\G_{(x,y)}(\varrho)}&=\ptr{ \hi_A}{\left(\sqrt{\Eo(x)} \otimes I_{\hik} \right) W \varrho \; W^* 
\left(\sqrt{\Eo(x)} \otimes I_{\hik} \right)} \nonumber  \\
&=\I_x(\varrho), 
\end{align}
where we have used $\ptr{\hv}{\sum_y \left( \R^{(x)}_y \otimes id_\hik \right)(\xi)}=\ptr{\hv}{\left(\Phi^{\R^{(x)}} \otimes id_\hik\right)(\xi)}=\ptr{\hi_A}{\xi}$ for all $\xi \in \mathcal{L}(\hi_A \otimes \hik)$. On the other hand, for all $\varrho \in \sh$ we have that
\begin{align}
\ptr{\hik}{\sum_x\G_{(x,y)}(\varrho)} &= \sum_x \R^{(x)}_y \left(\ptr{\hik}{\left(\sqrt{\Eo(x)} \otimes I_{\hik} \right) W \varrho  W^*  \left(\sqrt{\Eo(x)} \otimes I_{\hik} \right)}\right) \nonumber \\
&=\sum_x \R^{(x)}_y (\I^C_x(\varrho))=\J_y(\varrho). 
\end{align}
We conclude that $\G$ is a valid joint instrument for $\I$ and $\J$, so $\I \comp \J$.

Next, we want to prove that \textit{i)} implies \textit{ii)}. We do that by showing first that there exists a complementary instrument $\I^C$ related to some (minimal) dilation of $\I$ such that $\J \leq \I^C$. Then, from the postprocessing equivalence of all the complementary instruments that was shown in Prop. \ref{prop:ins-dil-pp-equiv} it follows that in fact $\J$ can be postprocessed from any complementary instrument of $\I$.

Let us fix some minimal dilation $(\hi_\I, W^\I, \Eo^\I)$ for $\I$, and let us consider some (not necessarily minimial) dilation $(\hi_\G, W^\G, \Eo^\G)$ of the joint instrument $\G \in \ins(\Omega \times \Lambda, \hi, \hv \otimes \hik)$ so that $W^\G: \hi \rightarrow \hi_\G \otimes \hv \otimes \hik$ and $\Eo^\G \in \obs(\Omega \times \Lambda, \hi_\G)$. We notice that $(\hi_\G \otimes \hv , W^\G)$ defines a dilation for the channel $\Phi^\I$ because
\begin{align}
\Phi^\I(\varrho)&=\sum_x \I_x(\varrho)=\sum_{x,y} \ptr{\hv}{\G_{(x,y)}(\varrho)} \nonumber \\
&=\sum_{x,y} \ptr{\hv}{\ptr{\hi_\G}{W^\G \;\varrho \; (W^\G)^* \left( \Eo^\G(x,y)\otimes I_{\hv}\otimes I_{\hik} \right)}} \nonumber\\
&=\ptr{\hi_\G \otimes \hv}{W^\G \;\varrho \; (W^\G)^*}
\end{align}
for all $\varrho \in \sh$. Any dilation of the channel $\Phi^\I$ is linked to its minimal dilation by an isometry, thus there exists an isometry $U:\hi_\I \to \hi_\G \otimes \hv$ such that
\begin{align}
\label{ref:relofIsom}
W^\G=(U \otimes I_{\hik}) W^\I
\end{align}

Next, we recall an observation already made in \cite{HeinosaariWolf10} that any joint POVM $\G\in \obs(\Omega \times \Lambda, \hi)$ for a pair of compatible POVMs $\A \in \obs(\Omega, \hi)$, $\B\in \obs(\Lambda, \hi)$ has a sequential implementation. More precisely, there exist a choice of Hilbert space $\hik$, an instrument $\Q \in \ins(\Omega, \hi, \hik)$ and a POVM $\C \in \obs(\Lambda, \hik)$ such that the concatenation of the instrument $\Q$ with the POVM $\C$ performs a joint POVM $\G$, i.e. $\Q^*_x(  \C(y)) = \G_{(x,y)}$ for all $x\in \Omega$, $y \in \Lambda$, and $\A^\Q = \A$, $(\Phi^{\Q})^*(\C)=\B$. Previous observation can be used also a bit differently than it is stated. Namely, every POVM $\G$ having a pair of outcomes (i.e. outcome set is a Cartesian product of two finite sets) can be considered as the joint POVM for its marginals. Then we see that such a POVM can be realized sequentially (if the connecting Hilbert space has sufficient dimension) in such a way that we discover one of the outcomes in each step. Thus, we can realize the POVM $\Eo^\G$ as a concatenation of an instrument $\Q \in \ins(\Omega,\hi_\G,\mathbb{C}^{|\Lambda|})$ and a POVM $\C \in \obs(\Lambda, \mathbb{C}^{|\Lambda|})$, i.e.
\begin{align}
\label{ref:2stepMeas}
\Eo^\G(x,y)=(\Q_x)^*\;(\C(y))
\end{align}
for all $x \in \Omega$ and $y \in \Lambda$.

In order to keep formulas short we will denote by caligraphic letters the quantum operations which perform the conjugation with the given isometry. For example, now we have $\mathcal{W}^\I(\varrho)= W^\I\,\varrho \,(W^\I)^*$, and $(U \otimes I_{\hik}) W^\I\;\varrho \;(W^\I)^*( U^* \otimes I_{\hik})=((\mathcal{U} \otimes id_\hik)\circ\mathcal{W}^\I)(\varrho)$, where $id_\hik$ is the identity channel on $\hik$.

Using Eqs. (\ref{ref:relofIsom}) and (\ref{ref:2stepMeas}) we can rewrite the action of the joint instrument $\G$ as
\begin{align}
\label{ref:rewriteInstG}
\G_{(x,y)}(\varrho)&=\ptr{\hi_\G}{\mathcal{W}^\G(\varrho)\left(\Eo^\G(x,y) \otimes I_{\hv}\otimes I_{\hik}\right)} \nonumber \\
&=\ptr{\hi_\G}{((\mathcal{U} \otimes id_\hik)\circ\mathcal{W}^\I)(\varrho) \left(\Q^*_x((\C(y)) \otimes I_{\hv}\otimes I_{\hik}\right)} \nonumber \\
&=\ptr{\mathbb{C}^{|\Lambda|}}{((\Q_x\otimes id_{\hv \otimes \hik})\circ(\mathcal{U} \otimes id_\hik)\circ\mathcal{W}^\I)(\varrho)\left(\C(y) \otimes I_{\hv}\otimes I_{\hik}\right)} \nonumber \\
&=\ptr{\mathbb{C}^{|\Lambda|}}{((\N_x\otimes id_\hik)\circ\mathcal{W}^\I)(\varrho)\left(\C(y) \otimes I_{\hv}\otimes I_{\hik}\right)}
\end{align}
for all $\varrho \in \sh$, $x \in \Omega$ and $y \in \Lambda$, where we defined the  instrument $\N \in \ins(\Omega,\hi_\I,\mathbb{C}^{|\Lambda|}\otimes \hv)$ as
\begin{align}
\label{ref:defInstN}
\N_x = (\Q_x \otimes id_{\hv})\circ \mathcal{U}
\end{align}
for all $x \in \Omega$. This implies that
\begin{align}
\label{ref:rewriteInstI}
\I_{x}(\varrho)&=\ptr{\mathbb{C}^{|\Lambda|} \otimes \hv}{((\N_x\otimes id_\hik)\circ\mathcal{W}^\I)(\varrho)\left(\left(\sum_y \C(y)\right) \otimes I_{\hv}\otimes I_{\hik}\right)} \nonumber\\
&=\ptr{\hi_\I}{\mathcal{W}^\I(\varrho)\left(\A^{\N}(x)\otimes I_{\hik}\right)}
\end{align}
for all $\varrho \in \sh$ and $x \in \Omega$. Thus, since $(\hi_\I,W^\I,\Eo^\I)$ is a minimal dilation, it follows from the uniqueness of the POVM in the minimal dilation discussed in Remark \ref{remark:min-dil-unique-povm} that $\A^\N = \Eo^\I$.

As a next step, we use the known result disussed in Example \ref{ex:luders} that every instrument can be realized as a concatenation of a Lüders instrument of its induced POVM and conditional postprocessing channels depending on the obtained outcome. In our situation it means that there exist channels $\{\E^{(x)} \}_{x \in \Omega} \subset \ch(\hi_\I,\mathbb{C}^{|\Lambda|}\otimes \hv)$ such that
\begin{align}
\label{ref:rewriteInstN}
\N_x = \E^{(x)} \circ \I^{\A^{\N}}_x  =  \E^{(x)} \circ \I^{\Eo^{\I}}_x \quad \forall x \in \Omega.
\end{align}
Using Eq. (\ref{ref:rewriteInstN}) we can rewrite Eq. (\ref{ref:rewriteInstG}) as
\begin{align}
\label{ref:rewriteInstG2}
\G_{(x,y)}(\varrho)&=\ptr{\mathbb{C}^{|\Lambda|}}{((\E^{(x)}\otimes id_\hik) \circ (\I^{\Eo^{\I}}_x\otimes id_\hik)
\circ\mathcal{W}^\I)(\varrho)\left(\C(y) \otimes I_{\hv}\otimes I_{\hik}\right)}
\end{align}
for all $\varrho \in \sh$, $x \in \Omega$ and $y \in \Lambda$,
and consequently
\begin{align}
\label{ref:rewriteInstJ}
\J_y(\varrho) &= \sum_x \ptr{\hik}{\G_{(x,y)}(\varrho)}  \nonumber \\
&=\sum_x \ptr{\mathbb{C}^{|\Lambda|}}
{\E^{(x)} \left( \ptr{\hik}{((\I^{\Eo^{\I}}_x\otimes id_\hik) \circ \mathcal{W}^\I)(\varrho)}\right) \left(\C(y) \otimes I_{\hv}\right)} \nonumber\\
&=\sum_x\ptr{\mathbb{C}^{|\Lambda|}}{
(\E^{(x)}\circ \I^C_x) (\varrho) \left(\C(y) \otimes I_{\hv}\right)}\nonumber\\
&=\sum_x(\mathcal{M}^{(x)}_y \circ \I^C_x)(\varrho)
\end{align}
for all $\varrho \in \sh$ and $y \in \Lambda$,
where we have defined instruments
$\mathcal{M}^{(x)}\in \ins(\Lambda,\hi_\I,\hv)$ as
\begin{equation*}
    \mathcal{M}^{(x)}_y(\xi) = \ptr{\mathbb{C}^{|\Lambda|}}{\E^{(x)} (\xi) \left(\C(y) \otimes I_{\hv}\right)}
\end{equation*}
for all $\xi \in \mathcal{L}(\hi_\I)$, $x \in \Omega$ and $y \in \Lambda$. Thus, the above Eq. (\ref{ref:rewriteInstJ}) shows that $\J \preceq \I^C$ for the complementary instrument $\I^C$ related to the minimal dilation $(\hi_\I, W^\I, \Eo^\I)$. As was mentioned before, now from Prop. \ref{prop:ins-dil-pp-equiv} it follows that in fact $\J \preceq \I^{C'}$ also for any other complementary instrument $\I^{C'}$ of $\I$. 
\end{proof}

In the next section we explore some of the consequences of Prop. \ref{prop:IncompAndPostproc}.

\subsection{Compatibility of a POVM and an instrument}
As was shown in Example \ref{ex:compl-ins-povm}, (one of) the explicit form(s) of the complementary instruments of a POVM $\A \in \obs(\Omega, \hi)$ is the Lüders instrument $\I^\A \in \ins(\Omega, \hi)$. Thus, from Proposition \ref{prop:IncompAndPostproc} we get a direct characterization for compatibility of a POVM and an instrument in terms of the instrument postprocessing relation including the Lüders instrument of the POVM:

\begin{corollary}\label{cor:ins-povm-comp}
Let $\A \in \obs(\Omega, \hi)$ be a POVM and $\J \in \ins(\Lambda,\hi, \hik)$ an instrument. Then $\A \comp \J$ if and only if $\J \preceq \I^\A$.
\end{corollary}

 We can see how Corollary \ref{cor:ins-povm-comp} can be seen as a generalization of the know result from Example \ref{ex:luders}: Namely, as we recall from Example \ref{ex:luders}, for an instrument $\I \in \ins(\Omega, \hi, \hik)$ and an observable $\A \in \obs(\Omega, \hi)$ we have that $\A^\I = \A$ if and only if $\I_x = \E^{(x)} \circ \I^\A_x$ for all $x \in \Omega$ for some channels $\{\E^{(x)}\}_{x \in \Omega} \subset \ch(\hi, \hik)$, and in particular then $\I \preceq \I^\A$. As was pointed out in Remark \ref{remark:A-comp}, if $\A^\I = \A$, then $\A$ and $\I$ are compatible. Thus, in particular the previous postprocessing relation involving only the conditional channels can be seen as a sufficient condition for compatibility between an instrument and a POVM while the most general postprocessing relation from Corrolary \ref{cor:ins-povm-comp} represents a necessary and sufficient condition for their compatibility. Alternatively, the above Corollary can be also interpreted as a characterization of those Lüders instruments, which can be postprocessed to a given instrument $\J$: these are exactly those Lüders instruments whose induced POVM is compatible with $\J$.

Naturally, since channels are one-outcome instruments, we also get a characterization for compatibility of a POVM $\A \in \obs(\Omega, \hi)$ and a channel $\Phi \in \ch(\hi,\hik)$: \emph{ $\A \comp \Phi$ if and only if $\Phi \preceq \I^\A$.} On the other hand, using just the postprocessing of channels, the authors in \cite{HeinosaariMiyadera13} showed that $\A \comp \Phi$ if and only if $\Phi \preceq \Psi^\A$, where $\Psi^\A \in \ch(\hi, \hv)$ is a particular (class of) channel(s) called the \emph{least disturbing channel of $\A$}. We will explore this connection a bit further.

The least disturbing channel can be expressed by using the  Naimark dilations of a POVM.

\begin{definition}
For an observable $\A \in \obs(\Omega, \hi)$ we say that a triple $(\hv, \Po, J)$, consisting of a Hilbert space $\hv$,  a sharp POVM $\Po \in \obs(\Omega, \hv)$ and an isometry $J: \hi \to \hv$, is a \emph{Naimark dilation for $\A$} if $\A(x) = J^* \Po(x) J$ for all $x \in \Omega$. The dilation is called \emph{minimal} if the closure of the span of the vectors $\{\Po(x)J\phi \, | \, x \in \Omega, \ \phi \in \hi\}$ is $\hv$.
\end{definition}

Now the least disturbing channel $\Psi \in \ch(\hi, \hv)$ of a POVM $\A \in \obs(\Omega, \hi)$ related to some Naimark dilation $(\hv, \Po, J)$ for $\A$ is defined as
\begin{align*}
\Psi^\A(\varrho) = \sum_{x \in \Omega} \Po(x)J \varrho J^* \Po(x)
\end{align*}
The authors in \cite{HeinosaariMiyadera13} showed that all least disturbing channels of a given POVM related to some Naimark dilation of that POVM are postprocessing equivalent so that the specific dilation is not relevant if one is considering postprocessing relations between channels. 

Thus, we have that $\Phi \comp \A$ if and only if $\Phi \preceq \Psi^\A$ if and only if $\Phi \preceq \I^\A$.  One of the obvious differences in the two previous postprocessing relations is that generally speaking the output space $\hv$ of the least disturbing channel is often different (and larger) than $\hi$ (which is the output space of $\I^\A$). In particular, the equivalence class of the least disturbing channels for any Naimark dilation for $\A$ has a natural minimal representative $\tilde{\Psi}^{\A}$ related to the minimal Naimark dilation for $\A$. One important feature of $\tilde{\Psi}^{\A}$ is that from all least disturbing channels of $\A$, $\tilde{\Psi}^{\A}$ has the minimal output dimension which equals $ \sum_{x \in \Omega} \mathrm{rank}(\A(x))$. It is easy to see that then one can choose $\hv= \hi$ if and only if $\A$ is sharp in which case one has $\tilde{\Psi}^{\A} = \Phi^{\I^\A}$. Thus, for a sharp POVM $\A \in \obs(\Omega, \hi)$ one has that a channel $\Phi \in \ch(\hi, \hik)$ is compatible with $\A$ if and only if $\Phi \preceq  \Phi^{\I^\A}$. Furthermore, in this case one can actually show (see \cite{Pellonpaa13}) that the postprocessing channel $\Phi' \in \ch(\hi,\hik)$ which satisfies $\Phi = \Phi' \circ \Phi^{\I^\A}$ is in fact unique and actually $\Phi' = \Phi$. As we saw in Section \ref{sec:4} this simply means that \emph{$\Phi$ is compatible with the sharp POVM $\A$ if and only if the Lüders instrument $\I^\A$ of $\A$ does not disturb $\Phi$.} 

As a last note we can show that actually $\I^\A$ and $\Psi^\A$ are postprocessing equivalent. Namely, if we define an instrument $\J \in \ins(\Omega, \hi, \hv)$ as $\J_x(\varrho) = \Po(x) J \varrho J^* \Po(x)$ for all $x \in \Omega$ and $\varrho \in \sh$, firstly see that $\J$ is clearly postprocessing equivalent with $\Psi^\A$ (as $\Psi^\A$ is just a classical postprocessing of $\J$ and since $\I^{\Po}_x \circ \Psi^\A = \J_x$ for all $x \in \Omega$), and secondly, since $\A^\J = \A$ and since both $\J$ and $\I^\A$ are indecomposable, by the results of \cite[Prop. 9]{LeppajarviSedlak21} we must have that also $\J$ and $\I^\A$ are postprocessing equivalent. Thus, this shows that our result Corollary \ref{cor:ins-povm-comp} in the case of channels and the result of \cite{HeinosaariMiyadera13} for compatibility of channels and POVMs are not only logically equivalent but also the details of their connection is clear. This also means that we can then see Corollary \ref{cor:ins-povm-comp} as a direct generalization of this result of \cite{HeinosaariMiyadera13} to instruments.

As the final immediate consequence of Prop. \ref{prop:IncompAndPostproc} we can look at the case where in Cor. \ref{cor:ins-povm-comp} also $\J$ is a POVM. Thus, we get the following necessary and sufficient condition for compatibility of two POVMs. 
\begin{corollary}\label{cor:povm-povm-comp}
Let $\A \in \obs(\Omega, \hi)$. Then a POVM $\B \in \obs(\Lambda, \hi)$ is compatible with $\A$ if and only if there exists some POVMs $\{\mathsf{R}^{(x)} \}_{x \in \Omega} \subset \obs(\Lambda,\hi)$ such that 
\begin{equation}
    \B(y) = \sum_{x \in \Omega} \sqrt{\A(x)}\mathsf{R}^{(x)}(y)\sqrt{\A(x)}
\end{equation}
for all $y \in \Lambda.$
\end{corollary}

\section{Simple classes of compatible instruments} \label{sec:6}

\subsection{Measure-and-prepare instruments and channels}\label{sec:m-a-p}
An instrument $\I \in \ins(\Omega, \hi, \hik)$ is said to be \emph{measure-and-prepare} if there exists a POVM $\A \in \obs(\Omega, \hi)$ and a family of states $\{\xi_x\}_{x \in \Omega} \subset \shik$ such that $\I_x(\varrho) = \tr{\A(x) \varrho} \xi_x$ for all $x \in \Omega$. Clearly then $\A^\I = \A$. The interpretation of measure-and-prepare instruments is evident: they measure the input state and output a state based on the measurement outcome. Similarly a channel is said to be measure-and-prepare if it is the induced channel of a measure-and-prepare instrument. It can be shown that measure-and-prepare channels are exactly those that are entanglement breaking \cite{HorodeckiShorRuskai03}. Special cases of measure-and-prepare instruments/channels include the trivial trash-and-prepare instruments/channels which measure just some trivial POVM whose effects are all proportional to the identity operator.

It is known (see e.g. \cite{HeinosaariMiyadera17,GirardPlavalaSikora21}) and it is straightforward to verify that if $\A \in \obs(\Omega, \hi)$ and $\B\in \obs(\Lambda, \hi)$ are two POVMs then, $\A \comp \B$ if and only if $\Phi_\A \comp \Phi_\B$ for some measure-and-prepare channels which measure $\A$ and $\B$, respectively, and which both prepare some sets of distinguishable states in some suitable output spaces where enough distinguishable states are available (which might be different from $\hi$). Thus, the compatibility between POVMs can always be rephrased as compatibility between some measure-and-prepare channels. Similar result is known also in the case of compatibility between a channel $\Phi \in \ch(\hi,\hik)$ and a POVM $\A \in \obs(\Omega, \hi)$ \cite{HeinosaariMiyadera17,GirardPlavalaSikora21}: \emph{If  $\Phi \comp \A$, then $\Phi \comp \Phi_\A$ for any measure-and-prepare channel $\Phi_\A$ which measures $\A$. Furthermore, if the states that $\Phi_\A$ prepares are distinguishable, then $\Phi \comp \Phi_\A$ implies that $\Phi \comp \A$.}

For instruments we see that the distinguishability of the prepared states is not needed. Indeed, for instruments we do not need to extract the measurement outcome by distinguishing the prepared states but it is automatically given to us as part of the description of the instrument. Thus, the compatibility of an instrument $\I$ with a measure-and-prepare instrument $\J$ reduces completely to compatibility between $\I$ and the POVM $\A^\J$. 

\begin{proposition}\label{prop:m-a-p-comp}
Let $\I \in \ins(\Lambda, \hi, \hik)$ be an instrument and $\B \in \obs(\Lambda, \hi)$ a POVM. Then $\I \comp \B$ if and only if $\I \comp \J^\B$ for any measure-and-prepare instrument $\J^\B$ that measures $\B$.
\end{proposition}
\begin{proof}
As we saw previously in Prop. \ref{prop:ins-comp-sum}, if $\I$ and $\J$ are compatible, then $\I$ is compatible with  $\A^\J$ even if $\J$ is not a measure-and-prepare instrument. To see that also the converse holds when $\J$ is a measure-and-prepare instrument, let $\I \in \ins(\Omega,\hi, \hik)$ and let $\J \in \ins(\Lambda,\hi, \hv)$ be of the form $\J_y(\varrho) = \tr{\B(y) \varrho} \xi_y$ for all $y \in \Lambda$ and $\varrho \in \sh$ for some family of states $\{\xi_y\}_{y \in \Lambda} \subset \shv$. If now $\I \comp \B$, so that there exists an instrument $\R \in \ins( \Omega \times \Lambda, \hi, \hik)$ such that $\sum_{y \in \Lambda} \R_{(x,y)} = \I_x$ for all $x \in \Omega$ and $\A^\R$ is a joint POVM for $\A^\I$ and $\A^\J=\B$, we can define an instrument $\G \in \ins(\Omega \times \Lambda, \hi, \hik \otimes \hv)$ by setting $\G_{(x,y)} = \R_{(x,y)} \otimes \xi_y$ for all $x \in \Omega$ and $y \in \Lambda$. It follows from the properties of $\R$ that $\G$ is a joint instrument of $\I$ and $\J$.
\end{proof}

By using Cor. \ref{cor:ins-povm-comp} and Prop. \ref{prop:m-a-p-comp} we get the following corollary:

\begin{corollary}
Let $\I \in \ins(\Lambda, \hi, \hik)$ be an instrument, $\B \in \obs(\Lambda, \hi)$ a POVM and $\J^\B \in \ins(\Lambda, \hi, \hv)$ some measure-and-prepare instrument which measures $\B$. Then $\I \comp \J^\B$ if and only if $\I \preceq \I^\B$.
\end{corollary}

Let us now consider the case of two measure-and-prepare instruments $\I \in \ins(\Omega, \hi, \hik)$ and $\J \in \ins(\Lambda, \hi, \hv)$. This means that there exists POVMs $\A, \B \in \obs(\Omega, \hi)$ and states $\{\sigma_x\}_{x \in \Omega} \subset \shik$, $\{\xi_y\}_{y \in \Lambda} \subset \shv$ such that $\I_x(\varrho) = \tr{\A(x) \varrho} \sigma_x$ and $\J_y(\varrho) = \tr{\B(y)\varrho} \xi_y$ for all $x \in \Omega$, $y \in \Lambda$ and $\varrho \in \sh$. As mentioned previously, then $\A^\I = \A$ and $\A^\J = \B$, and both $\Phi^\I$ and $\Phi^\J$ are measure-and-prepare channels.

The following two results are previously known for measure-and-prepare channels (see e.g. \cite{GirardPlavalaSikora21}):
\begin{itemize}
\item If $\A^\I \comp \A^\J$, then $\Phi^\I \comp \Phi^\J$.
\item If $\Phi^\I \comp \Phi^\J$ and the both sets of states $\{\sigma_x\}_{x \in \Omega}$ and $\{\xi_y\}_{y \in \Lambda}$ are distinguishable (independently of each other), then $\A^\I \comp \A^\J$.
\end{itemize}

The first result can be easily seen as follows: if $\A^\I \comp \A^\J$, then there exists a joint POVM $\Go \in \obs(\Omega \times \Lambda, \hi)$ for $\A^\I$ and $\A^\J$ and we can define a joint channel $\Gamma \in \ch(\hi, \hik \otimes \hv)$ for $\Phi^\I$ and $\Phi^\J$ by setting $\Gamma(\varrho) = \sum_{x \in \Omega} \sum_{y \in \Lambda} \tr{\Go(x,y) \varrho} \sigma_x \otimes \xi_y$ for all $\varrho \in \sh$. Clearly then $\ptr{\hik}{\Gamma(\varrho)} = \Phi^\J(\varrho)$ and $\ptr{\hv}{\Gamma(\varrho)} = \Phi^\I(\varrho)$ for all $\varrho \in \sh$. From this we see that if we define $\G \in \ins(\Omega \times \Lambda, \hi, \hik \otimes \hv)$ by setting $\G_{(x,y)}(\varrho) = \tr{\Go(x,y) \varrho} \sigma_x \otimes \xi_y$ for all $\varrho \in \sh$, $x \in \Omega$ and $y \in \Lambda$, then $\I \comp \J$ with $\G$ being their joint instrument. Hence, together with Prop. \ref{prop:ins-comp-sum} we conclude the following result:
\begin{proposition}\label{prop:m-a-p-comp1}
Two measure-and-prepare instruments $\I$ and $\J$ are compatible if and only if their induced POVMs $\A^\I$ and $\A^\J$ are compatible.
\end{proposition}

The previously stated observation shows that for measure-and-prepare instruments the compatibility of the instruments is purely dictated by the compatibility of the measured POVMs. The next example shows that the same does not hold for measure-and-prepare channels, i.e., there exists an incompatible pair of measure-and-prepare instruments whose induced channels are compatible.

\begin{example}
We recall another useful representation of quantum channels and operations \emph{the Choi-Jamiołkowski isomorphism} \cite{Choi75}. Quantum operation $\N: \lh \to \lk$ is represented by a positive linear operator $J(\N)$ on $\mathcal{L}(\hi \otimes \hik)$ defined as
\begin{align*}
J(\N) := \sum_{i,j=1}^{\dim(\hi)} \kb{\varphi_i}{\varphi_j} \otimes \N(\kb{\varphi_i}{\varphi_j})
\end{align*}
for some orthonormal basis $\{\varphi_i\}_{i=1}^{\dim(\hi)}$ of $\hi$. It can be shown that $\N$ is completely positive if and only if $J(\N) \geq O$. Furthermore, $\N$ is trace-preserving, i.e. $\N$ is a quantum channel, if and only if $\ptr{\hik}{J(\N)} = I_{\hi}$.

The Jordan product of two channels $\Phi \in \ch(\hi, \hik)$ and $\Psi \in \ch(\hi, \hv)$ is defined \cite{GirardPlavalaSikora21} as the linear map $\Phi \odot \Psi$ from $\hi$ to  $\hik \otimes \hv$ whose Choi representation is the operator
\begin{align*}
J(\Phi \odot \Psi) := \sum_{i,j,k,l=1}^{\dim(\hi)} \left( \kb{\varphi_i}{\varphi_j} \odot \kb{\varphi_k}{\varphi_l} \right) \otimes \Phi(\kb{\varphi_i}{\varphi_j}) \otimes \Psi(\kb{\varphi_k}{\varphi_l}),
\end{align*}
where $\{\varphi_i\}_{i=1}^{\dim(\hi)}$ is an orthonormal basis of $\hi$ and $A \odot B := \frac{1}{2}(AB+BA)$ is the standard Jordan product of operators on $\hi$. It can be checked that $\ptr{\hik}{(\Phi \odot \Psi)(\varrho)} = \Psi(\varrho)$ and $\ptr{\hv}{(\Phi \odot \Psi)(\varrho)} = \Phi(\varrho)$ for all $\varrho \in \sh$. Thus, $\Phi \odot \Psi$ is a joint channel for $\Phi$ and $\Psi$ if and only if it is completely positive, i.e., if $J(\Phi \odot \Psi)$ is positive semi-definite.

Let $\hi = \complex^2$ and define two dichotomic POVMs $\X, \Z \in \obs(\{+,-\}, \complex^2)$ as $\X(\pm) = \frac{1}{2}(I \pm \sigma_X)$ and $\Z(\pm) = \frac{1}{2}(I \pm \sigma_Z)$, where $\sigma_X$ and $\sigma_Z$ are the Pauli-$X$ and Pauli-$Z$ matrices. It is known that $\X$ and $\Z$ are incompatible (in fact they are the most incompatible pair of POVMs on $\complex^2$).  For each $a \in [0,1]$, let us consider the states $\varrho^{\pm X}_{a} = \frac{1}{2}(I \pm a \sigma_X)$ and $\varrho^{\pm Z}_a = \frac{1}{2}(I \pm a \sigma_Z)$ in $\mathcal{S}(\complex^2)$, and define two measure-and-prepare instruments $\I^a, \J^a \in \ins(\{+,-\}, \complex^2)$ by setting $\I^a_\pm(\varrho) = \tr{ \X(\pm) \varrho } \varrho^{\pm X}_a  $ and $\J^a_\pm(\varrho) = \tr{ \Z(\pm) \varrho } \varrho^{\pm Z}_a  $ for all $\varrho \in \mathcal{S}(\complex^2)$. It can be checked (via semidefinite programming) that $J(\Phi^{\I^a} \odot \Phi^{\J^a}) \geq 0$ if and only if $a \leq 1/\sqrt{2}$. Thus, for all $a \in [0,1/\sqrt{2}]$ we have that $\Phi^{\I^a} \comp \Phi^{\J^a}$ but since $\A^{\I^a} = \X$ and $\A^{\J^a} = \Z$ are incompatible, also $\I^a$ and $\J^a$ are incompatible. 
\end{example}

\subsection{Indecomposable instruments}
We will show that indecomposable instruments process input quantum states in such a way that no quantum information can flow to the output of the instruments that are compatible with them. More precisely, we will use Prop. \ref{prop:IncompAndPostproc} to show that instruments compatible with an indecomposable instrument must be a postprocessing of a measure-and-prepare instrument. First, let us start by finding a dilation for an indecomposable instrument that we can use:

\begin{lemma} \label{lemma:ind-dilation}
    For an indecomposable instrument $\I \in \ins(\Omega, \hi, \hik)$ there exists a dilation $(\hi_A, W, \Eo)$ such that the complementary instrument $\I^C \in \ins(\Omega, \hi, \hi_A)$ is a measure-and-prepare instrument, 
    \begin{equation}\label{eq:m-a-p-ind}
        \I^C_x(\varrho) = \tr{\A^\I(x) \varrho} \kb{\varphi_{x}}{\varphi_{x}} 
    \end{equation}
    for all $x \in \Omega$ and $\varrho \in \sh$ for some orthonormal basis  $\{\varphi_x\}_{x \in \Omega}$ of $\hi_A$.
\end{lemma}
\begin{proof}
    Let $\I \in \ins(\Omega, \hi, \hik)$ be indecomposable so that there exist Kraus operators $\{K_x\}_{x \in \Omega}$ from $\hi$ to $\hik$ such that $\I_x(\varrho) = K_x \varrho K^*_x$ for all $x \in \Omega$ and $\varrho \in \sh$. Then clearly $\A^\I(x) = K^*_x K_x$ for all $x \in \Omega$. Let us fix any orthonormal basis $\{\varphi_x\}_{x \in \Omega}$ of $\hi_A := \complex^{|\Omega|}$ and define $W: \hi \to \hi_A \otimes \hik$ by setting $W \psi = \sum_{x \in \Omega} \varphi_x \otimes K_x \psi$ for all $\psi \in \hi$. It is straightforward to verify that $W^*W=I_\hi$ so that $W$ is an isometry. Finally, we define a POVM $\Eo \in \obs(\Omega, \hi_A)$ by taking $\Eo(x) = \kb{\varphi_x}{\varphi_x}$ for all $x \in \Omega$. One sees that $  W \varrho \, W^*  = \sum_{x',x'' \in \Omega} \kb{\varphi_{x'}}{\varphi_{x''}} \otimes K_{x'} \varrho K_{x''}^*$ for all $\varrho \in \sh$ so that in particular $\ptr{\complex^{|\Omega|}}{W \varrho \, W^* (\Eo(x) \otimes I_\hik)} =  K_x \varrho K^*_x = \I_x(\varrho)$ for all $x \in \Omega$ and $\varrho \in \sh$. Thus, $(\hi_A, W ,\Eo)$ is a dilation of $\I$. A direct calculation shows that
\begin{align*}
   \I^C_x(\varrho) &= \ptr{\hik}{\left(\sqrt{\Eo(x)} \otimes I_\hik\right) W \varrho \, W^* \left(\sqrt{\Eo(x)} \otimes I_\hik\right)}   \\
   &= \ptr{\hik}{\kb{\varphi_{x}}{\varphi_{x}} \otimes K_{x} \varrho K_{x}^*}   \\
   &= \tr{\A^\I(x) \varrho} \kb{\varphi_{x}}{\varphi_{x}} 
\end{align*}
for all $x \in \Omega$ and $\varrho \in \sh$.
\end{proof}

By using the above dilation, from Prop. \ref{prop:IncompAndPostproc} we now see that instruments compatible with an indecomposable instrument must be a postprocessing of the measure-and-prepare instrument described above. We get the following Corollary:

\begin{corollary}\label{cor:ind-comp-m-a-p}
Let $\I \in \ins(\Omega, \hi, \hik)$ be an indecomposable instrument. Then an instrument $\J \in \ins(\Lambda, \hi, \hv)$ that is compatible with $\I$ must be of the form
\begin{equation}\label{eq:ind-m-a-p}
    \J_y(\varrho) = \sum_{x \in \Omega} \tr{ \nu_{xy} \A^{\I}(x)\varrho} \xi_{xy} 
\end{equation}
for some family of states $\{\xi_{xy}\}_{x \in \Omega, y \in \Lambda} \subset \shv$ and some stochastic postprocessing matrix $\nu:= (\nu_{xy})_{x \in \Omega, y \in \Lambda}$.
\end{corollary}
\begin{proof}
    From Prop. \ref{prop:IncompAndPostproc} we see that an instrument $\J \in \ins(\Lambda, \hi, \hv)$ is compatible with an indecomposable instrument  $\I \in \ins(\Omega, \hi, \hik)$ if and only if $\J \preceq \I^C$, where $\I^C$ can be chosen to be the measure-and-prepare instrument given in Eq. \eqref{eq:m-a-p-ind} from Lemma \ref{lemma:ind-dilation}. The rest follows similarly to Example 4 in \cite{LeppajarviSedlak21}. In particular, there exists $\{\R^{(x)}\}_{x \in \Omega} \subset \ins(\Lambda,\hi_A,\hv)$ such that 
    \begin{align*}
        \J_y(\varrho) = \sum_{x \in \Omega} \R^{(x)}_y(\I^C_x(\varrho)) = \sum_{x \in \Omega}\tr{\A^\I(x) \varrho} \R^{(x)}_y(\kb{\varphi_{x}}{\varphi_{x}})
    \end{align*}
    for all $y \in \Lambda$ and $\varrho \in \sh$. If we denote $\nu_{xy} = \tr{\R^{(x)}_y(\kb{\varphi_{x}}{\varphi_{x}})}  \in [0,1]$ we see that $\nu = (\nu_{xy})_{x,y}$ forms a stochastic postprocessing matrix. Furthermore, we can set $\xi_{xy} := \R^{(x)}_y(\kb{\varphi_{x}}{\varphi_{x}})/\nu_{xy}$ if $\nu_{xy} \neq 0$ and $\xi_{xy} := \xi$ for some fixed $\xi \in \shv$ when $\nu_{xy} = 0$. The claim follows.
\end{proof}
An important point to note is that $\A^\J(y) = \sum_{x \in \Omega} \nu_{xy} \A^\I(x)$ for all $y \in \Lambda$ so that for instruments $\J$ compatible with $\I$ we must have that $\A^\J$ is actually a postprocessing of $\A^\I$.

As an easy corollary of Prop. \ref{prop:IncompAndPostproc} and  Cor. \ref{cor:ind-comp-m-a-p} we present the case when $\J$ in Cor. \ref{cor:ind-comp-m-a-p} is a POVM, i.e., an instrument with one-dimensional output space as in Example \ref{ex:compl-ins-povm}.
\begin{corollary}\label{cor:ind-comp-povm}
Let $\I \in \ins(\Omega, \hi, \hik)$ be an indecomposable instrument. Then a POVM $\B \in \obs(\Lambda, \hi)$ is compatible with $\I$ if and only if $\B \preceq \A^\I$.
\end{corollary}
\begin{proof}
    As was noted above after Cor. \ref{cor:ind-comp-m-a-p}, if $\I$ is indecomposable and $\I \comp \J$ for any instrument  $\J$, then $\A^\J \preceq \A^\I$. On the other hand, if $\B \preceq \A^\I$ via some stochastic postprocessing matrix $\nu = (\nu_{xy})_{x \in \Omega, y \in \Lambda}$, then we can define a joint instrument $\G \in \ins(\Omega \times \Lambda, \hi, \hik)$ for $\I$ and $\B$ as $\G_{(x,y)} = \nu_{xy} \I_x$ for all $x \in \Omega$ and $y \in \Lambda$.
\end{proof}

Thus, an indecomposable instrument can only be compatible with those POVMs that are below its induced POVM in the postprocessing order.

Lastly we want to see when two indecomposable instruments are compatible. In order to characterize that we need the following lemma:

\begin{lemma}\label{lemma:ind-map}
An indecomposable instrument is measure-and-prepare if and only if its induced POVM is rank-1, i.e., it consists of rank-1 effects. In this case the prepared states are also pure.
\end{lemma}
\begin{proof}
Let $\I \in \ins(\Omega, \hi, \hik)$ be an indecomposable instrument so that it can be written as $\I_x(\varrho) = K_x \varrho K^*_x$ for some Kraus operators $K_x: \hi \to \hik$ for all $x \in \Omega$. Now clearly $\A^\I(x) = K_x^* K_x  $ for all $x \in \Omega$. Let first $\A^\I$ be rank-1. In general it is known (see \cite{HeinosaariWolf10}) that instruments with a rank-1 induced POVMs are measure-and-prepare instruments and since $\I$ is indecomposable, the prepared states must be pure states.

On the other hand, let now $\I$ be measure-and-prepare. Thus, we must have that there exists a family of states $\{\xi_x\}_{x \in \Omega} \subset \shik$ such that $\I_x(\varrho) = K_x \varrho K^*_x = \tr{K^*_x K_x \varrho} \xi_x$ for all $x \in \Omega$ and $\varrho \in \sh$. This implies that the range of $K_x$ is $ \complex \psi_x$, that is, $K_x = \ket{\psi_x}\bra{\eta_x}$ for some vector $\eta_x \in \hi$ for all $x \in \Omega$. Now $\A^\I (x) = K^*_x K_x =||\psi_x ||^2 \ket{\eta_x}\bra{\eta_x}$ for all $x \in \Omega$ so that $\A^\I$ is rank-1.
\end{proof}

By using the above lemma, together with Cor. \ref{cor:ind-comp-m-a-p} we can show that two indecomposable instruments are compatible only when the induced POVMs are essentially the same rank-1 POVM in which case the instruments must actually be measure-and-prepare instruments.

\begin{proposition}
Two indecomposable instruments are compatible if and only if their induced POVMs are postprocessing equivalent rank-1 POVMs. Furthermore, in this case both of the instruments are measure-and-prepare instruments which prepare pure states.
\end{proposition}
\begin{proof}
    Let $\I \in \ins(\Omega, \hi, \hik)$ and $\J \in \ins(\Lambda, \hi, \hv)$ be two indecomposable instruments. Let first $\A^\I$ and $\A^\J$ be postprocessing equivalent rank-1 POVMs. From Lemma \ref{lemma:ind-map} we see that $\I$ and $\J$ are measure-and-prepare instruments. Moreover, the induced POVMs $\A^\I$ and $\A^\J$ are compatible since they are postprocessing equivalent. Thus, by Prop. \ref{prop:m-a-p-comp1} instruments $\I$ and $\J$ are compatible.

    On the other hand, let now $\I$ and $\J$ be compatible. From Cor. \ref{cor:ind-comp-m-a-p} we have that $\J$ must be of the form given by Eq. \eqref{eq:ind-m-a-p}. Given the details of the proof of Cor. \ref{cor:ind-comp-m-a-p}, from the indecomposability of $\J$ it follows that $\xi_{xy} = \xi_{x'y} =: \xi_y$ for all $x, x' \in \Omega$ and $y \in \Lambda$ so that $\J$ is a measure-and-prepare instrument explicitly given as
    \begin{equation}
        \J_y(\varrho) = \tr{\sum_{x \in \Omega} \nu_{xy} \A^\I(x) \varrho} \xi_y
    \end{equation}
    for all $y \in \Lambda$ and $\varrho \in \sh$. Furthermore, the indecomposability of $\J$ also implies that the prepared states $\{\xi_y\}_{y \in \Lambda} $ must be rank-1, i.e., pure states, and that the measured POVM $\A^\J$ must also be rank-1 (as otherwise one could decompose the operations $\J_y$ into a nontrivial sum of other operations). As noted before, $\A^\J$ is a postprocessing of $\A^\I$, and since rank-1 POVMs are postprocessing maximal (or postprocessing clean) \cite{MartensMuynck90}, meaning that all POVMs which can be used to postprocess a postprocessing maximal POVM from are postprocessing equivalent with the maximal POVM.  It follows that actually $\A^\J$ must be postprocessing equivalent with $\A^\I$, i.e., $\A^\I \preceq \A^\J$ and $\A^\J \preceq \A^\I$. We note that two POVMs are postprocessing equivalent if and only if all of the effects of one of them are proportional to the effects of the other one (see e.g. \cite{LeppajarviSedlak21}). Thus, it follows that the postprocessing equivalence class of any rank-1 POVM only contains rank-1 POVMs so that in particular $\A^\I$ must also be rank-1. From Lemma \ref{lemma:ind-map} it then follows that $\I$ must be a measure-and-prepare instrument, which prepares only pure state and the claim follows.
\end{proof}

\section{Conclusions}\label{sec:7}
Here we would like to sketch two directions in which we hope our results could be extended further. Let us first return to channel compatibility. Joint channel for a pair of compatible channels can be seen as an object, which redistributes quantum information from its input into its two outputs and possibly also erases part of it. Complementary channel to a given channel can be thus understood as a specification of that part of quantum information, which can be retained, while creating the given channel. Similar view can be adopted also for POVMs. Namely, any joint POVM for a pair of compatible POVMs can be realized as a channel that distributes quantum state on the input into bipartite quantum system, whose parts are measured by a pair of independent measurements, each of them determining one of the outcomes of the joint POVM. 

Finally, from the proof of Proposition \ref{prop:IncompAndPostproc} one can conclude that also for compatible instruments there always exists such a realization of the joint instrument, where first step is a quantum channel distributing quantum information into two-partite output followed by individual instruments each producing one of the outcomes of the joint instrument. From this point of view compatibility is an interplay of what features of quantum states can be simultaneously distributed (broadcasted) to the multiple outputs. We remind that such distributing quantum channel might be partly (or even fully) a measure-and-prepare channel, which means that some "measured" features can be distributed simultaneously to all outputs. Morever, several distributing channels might lead to the realization of the same joint device. This is why it seems more practical than the above approach to characterize pairwise compatibility using complementary instruments, which we introduced in this manuscript. 

For a quantum channel its complementary instrument coincides with the notion of complementary channel. For a POVM one of the complementary instruments is the Lüders instrument defined by the same POVM.  Our Proposition \ref{prop:IncompAndPostproc} claims that if we want to find all devices that are compatible with a fixed device (POVM, channel, instrument) it suffice to characterize all devices resulting from the (instrument) postprocessing of its complementary quantum instrument. We believe that this approach can bring more intuitive understanding even into the compatibility of two POVMs. More precise statement in this direction is presented  in our Corollary \ref{cor:povm-povm-comp} and it essentially states that each element of a compatible POVM is a sum of (positive-semidefinite) pieces into which, the original POVM elements were chopped. While the above statement can be easily inferred directly from the definition of the joint measurement, Corollary \ref{cor:povm-povm-comp} gives one particular way of measuring the compatible POVMs in a sequential way by using the same Hilbert space unlike in the case of using the least disturbing channel as in \cite{HeinosaariMiyadera13}.

Second direction in which it might be interesting to investigate incompatibility of instruments is how to turn such incompatibility into a resource. Thus, one would look for some problem or suitably chosen optimization task in which a given  pair of incompatible instruments would perform strictly better than any compatible pair. This has been already considered in the general case of channels with mixed quantum-classical outputs in \cite{CarmeliHeinosaariMiyaderaToigo19} where it is shown that incompatibility gives a certain advantage in particular quantum state discrimination tasks. We believe that considering the same problem in the quantum instrument formulation could bring further intuition in identifying such tasks where an advantage emerges.

\section*{Acknowledgments}
We want to thank Teiko Heinosaari for the fruitful discussions around the topic. We also thank the anonymous referees for their suggestions and comments which greatly helped to improve the presentation of the manuscript. L.L. acknowledges that his part of the project has received funding from the European Union's Horizon 2020 Research and Innovation Programme under the Programme SASPRO 2 COFUND Marie Sklodowska-Curie grant agreement No. 945478. M.S. and L.L. were supported by projects APVV22-0570 (DeQHOST) and VEGA 2/0183/21 (DESCOM). M.S. was further supported by funding from QuantERA, an ERA-Net cofund in Quantum Technologies, under the project eDICT.


\begin{thebibliography}{18}
\bibitem{Heisenberg27}
W. Heisenberg, Uber den anschaulichen Inhalt der quantentheoretischen Kinematik und Mechanik, \href{http://dx.doi.org/10.1007/BF01397280}{Z. Phys. \textbf{43}, 172--198 (1927).}

\bibitem{Bohr28}
N. Bohr, The quantum postulate and the recent development of atomic theory, \href{https://doi.org/10.1038/121580a0}{Nature \textbf{121}, 580--590 (1928).}

\bibitem{GuhneHaapasaloKraftPellonpaaUola21}
O. Gühne, E. Haapasalo, T. Kraft, J.-P. Pellonpää, and R. Uola, \textit{Colloquium}: Incompatible measurements in quantum information science, \href{https://doi.org/10.1103/RevModPhys.95.011003}{Rev. Mod. Phys. \textbf{95}, 011003 (2023).}

\bibitem{Fine82}
A. Fine, Hidden variables, joint probability, and the bell inequalities, \href{https://doi.org/10.1103/PhysRevLett.48.291}{Phys. Rev. Lett. \textbf{48}, 291--295 (1982).}

\bibitem{WolfPerez-GarciaFernandez09}
M.M. Wolf, D. Perez-Garcia, and C. Fernandez, Measurements incompatible in quantum theory cannot be measured jointly in any other no-signaling theory, \href{https://doi.org/10.1103/PhysRevLett.103.230402}{Phys. Rev. Lett. \textbf{103}, 230402 (2009).}

\bibitem{LiangSpekkensWiseman11}
Y.-C. Liang, R. W. Spekkens, and H. M. Wiseman, Specker’s parable of the overprotective seer: A road to contextuality, nonlocality and complementarity, \href{https://doi.org/10.1016/j.physrep.2011.05.001}{Phys. Rep. \textbf{506}, 1--39 (2011).}

\bibitem{XuCabello19}
Z.-P. Xu and A. Cabello, Necessary and sufficient condition for contextuality from incompatibility, \href{https://doi.org/10.1103/PhysRevA.99.020103}{Phys. Rev. A \textbf{99}, 020103(R) (2019).}

\bibitem{TavakoliUola20}
A. Tavakoli and R. Uola, Measurement incompatibility and steering are necessary and sufficient for operational contextuality, \href{https://doi.org/10.1103/PhysRevResearch.2.013011}{Phys. Rev. Research \textbf{2}, 013011 (2020).}

\bibitem{QuintinoVertesiBrunner14}
M. T. Quintino, T. Vértesi, and N. Brunner, Joint Measurability, Einstein-Podolsky-Rosen Steering, and Bell Nonlocality, \href{https://doi.org/10.1103/PhysRevLett.113.160402}{Phys. Rev. Lett. \textbf{113}, 160402 (2014).}

\bibitem{SkrzypczykSupicCavalcanti19}
P. Skrzypczyk,  I. Šupić, and D. Cavalcanti, All Sets of Incompatible Measurements give an Advantage in Quantum State Discrimination, \href{https://doi.org/10.1103/PhysRevLett.122.130403}{Phys. Rev. Lett. \textbf{122}, 130403 (2019). }

\bibitem{CarmeliHeinosaariToigo19}
C. Carmeli, T. Heinosaari, and A. Toigo, Quantum Incompatibility Witnesses, \href{https://doi.org/10.1103/PhysRevLett.122.130402}{Phys. Rev. Lett. \textbf{122}, 130402 (2019).}

\bibitem{UolaKraftShangYuGuhne19}
R. Uola, T. Kraft, J. Shang, X.-D. Yu, and O. Gühne, Quantifying Quantum Resources with Conic Programming, \href{https://doi.org/10.1103/PhysRevLett.122.130404}{Phys. Rev. Lett. \textbf{122}, 130404 (2019).}

\bibitem{CarmeliHeinosaariToigo20}
C. Carmeli, T. Heinosaari, and A. Toigo, Quantum random access codes and incompatibility of measurements, \href{https://doi.org/10.1209/0295-5075/130/50001}{EPL \textbf{130}, 50001 (2020).}

\bibitem{HeinosaariLeppajarvi22}
T. Heinosaari and L. Leppäjärvi, Random access test as an identifier of nonclassicality, \href{https://doi.org/10.1088/1751-8121/ac5b91}{J. Phys. A: Math. Theor. \textbf{55}, 174003 (2022).}

\bibitem{Plavala16}
M. Plávala, All measurements in a probabilistic theory are compatible if and only if the state space is a simplex, \href{https://doi.org/10.1103/PhysRevA.94.042108}{Phys. Rev. A \textbf{94}, 042108 (2016).}

\bibitem{HeinosaariMiyaderaZiman16}
T. Heinosaari, T. Miyadera, and M. Ziman, An invitation to quantum incompatibility, \href{https://doi.org/10.1088/1751-8113/49/12/123001}{J. Phys. A: Math. Theor. \textbf{49}, 123001 (2016).}

\bibitem{HeinosaariMiyadera17}
T. Heinosaari and T. Miyadera, Incompatibility of quantum channels, \href{https://doi.org/10.1088/1751-8121/aa5f6b}{J. Phys. A: Math. Theor. \textbf{50}, 135302 (2017).}

\bibitem{Haapasalo19}
E. Haapasalo, Compatibility of covariant quantum channels with emphasis on Weyl symmetry, \href{https://doi.org/10.1007/s00023-019-00827-x}{Ann. Henri Poincaré \textbf{20}, 3163--3195 (2019).}

\bibitem{Kuramochi18}
Y. Kuramochi, Quantum incompatibility of channels with general outcome operator algebras, \href{https://doi.org/10.1063/1.5008300}{J. Math. Phys. \textbf{59}, 042203 (2018).}

\bibitem{GirardPlavalaSikora21}
M. Girard, M. Plávala, and J. Sikora, Jordan products of quantum channels and their compatibility, \href{https://doi.org/10.1038/s41467-021-22275-0}{Nat. Comm. \textbf{12}, 2129 (2021).}

\bibitem{CarmeliHeinosaariMiyaderaToigo19}
C. Carmeli, T. Heinosaari, T. Miyadera, A. Toigo, Witnessing incompatibility of quantum channels, \href{https://doi.org/10.1063/1.5126496}{J. Math. Phys. \textbf{60}, 122202 (2019).}

\bibitem{MitraFarkas22}
A. Mitra and M. Farkas, Compatibility of quantum instruments, \href{https://doi.org/10.1103/PhysRevA.105.052202}{Phys. Rev. A \textbf{105}, 052202 (2022).}

\bibitem{JiChitambar21}
K. Ji and E. Chitambar, Incompatibility as a resource for programmable quantum instruments, \href{https://doi.org/10.48550/arXiv.2112.03717}{arXiv:2112.03717 [quant-ph] (2021).}

 \bibitem{DArianoPerinottiTosini22}
G. M. D'Ariano, P. Perinotti, and A. Tosini, Incompatibility of observables, channels and instruments in information theories, \href{https://doi.org/10.1088/1751-8121/ac88a7}{J. Phys. A: Math. Theor. \textbf{55} 394006 (2022).}

\bibitem{MitraFarkas22b}
A. Mitra and M. Farkas, Characterizing and quantifying the incompatibility of quantum instruments, \href{https://doi.org/10.1103/PhysRevA.107.032217}{Phys. Rev. A \textbf{107}, 032217 (2023).}

\bibitem{BuscemiKobayashiMinagawaPerinottiTosini22}
F. Buscemi, K. Kobayashi, S. Minagawa, P. Perinotti, and A. Tosini, Unifying different notions of quantum incompatibility into a strict hierarchy of resource theories of communication, \href{	https://doi.org/10.22331/q-2023-06-07-1035}{Quantum \textbf{7}, 1035 (2023).}

\bibitem{HeinosaariWolf10}
T. Heinosaari and M. M. Wolf, Nondisturbing quantum measurements, \href{https://doi.org/10.1063/1.3480658}{J. Math. Phys. \textbf{51}, 092201 (2010).}

\bibitem{LeppajarviSedlak21}
L. Leppäjärvi and M. Sedlák, Postprocessing of quantum instruments, \href{https://doi.org/10.1103/PhysRevA.103.022615}{Phys. Rev. A \textbf{103}, 022615 (2021).}

\bibitem{KrausBook83}
K. Kraus, \textit{States, Effects, and Operations} \href{https://doi.org/10.1007/3-540-12732-1}{(Springer-Verlag, Berlin, 1983).}

\bibitem{Stinespring55}
W. F. Stinespring, Positive functions on {$C\sp *$}-algebras, \href{https://doi.org/10.2307/2032342}{Proc. Amer. Math. Soc. \textbf{6}, 211--216 (1955).}

\bibitem{HayashiBook06}
M. Hayashi, \textit{Quantum Information} \href{https://doi.org/10.1007/978-3-662-49725-8}{(Springer-Verlag, Berlin, 2006).} Translated from the 2003 Japanese original.

\bibitem{BarnumCavesFuchsJozsaSchumacher96}
H. Barnum, C. M. Caves, C. A. Fuchs, R. Jozsa, and B. Schumacher, Noncommuting mixed states cannot be broadcast, \href{https://doi.org/10.1103/PhysRevLett.76.2818}{Phys. Rev. Lett. \textbf{76}, 2818--21 (1996).}

\bibitem{WoottersZurek82}
W. Wootters and W. Zurek, A Single Quantum Cannot be Cloned, \href{https://doi.org/10.1038/299802a0}{Nature \textbf{299}, 802--803 (1982).}

\bibitem{HeinosaariReitznerStano08}
T. Heinosaari, D. Reitzner, and P. Stano, Notes on Joint Measurability of Quantum Observables, \href{https://doi.org/10.1007/s10701-008-9256-7}{Found. Phys. \textbf{38}, 1133--1147 (2008).}

\bibitem{AliCarmeliHeinosaariToigo09}
S. T. Ali, C. Carmeli, T. Heinosaari, and A. Toigo, Commutative POVMs and Fuzzy Observables, \href{https://doi.org/10.1007/s10701-009-9292-y}{Found. Phys. \textbf{39}, 593--612 (2009).}

\bibitem{LeppajarviThesis}
L. Leppäjärvi, Measurement simulability and incompatibility in quantum theory and other operational theories, \href{https://urn.fi/URN:ISBN:978-951-29-8476-30}{Annales Universitatis Turkuensis, Ser A I: 646 (PhD thesis, University of Turku, 2021).}

\bibitem{HeinosaariMiyadera13}
T. Heinosaari and T. Miyadera, Qualitative noise-disturbance relation for quantum measurements, \href{https://doi.org/10.1103/PhysRevA.88.042117}{Phys. Rev. A \textbf{88}, 042117 (2013).}

\bibitem{Jencova21}
A. Jencová, A general theory of comparison of quantum channels (and beyond), \href{https://doi.org/10.1109/TIT.2021.3070120}{IEEE Trans. Inf. Theory \textbf{67}, 3945--3964 (2021).}

\bibitem{BenyOreshkov11}
C. Bény and O. Oreshkov, Approximate simulation of quantum channels, \href{https://doi.org/10.1103/PhysRevA.84.022333}{Phys. Rev. A \textbf{84}, 022333 (2011).}

\bibitem{Raginsky03}
M. Raginsky, Radon–Nikodym derivatives of quantum operations, \href{https://doi.org/10.1063/1.1615697}{J. Math. Phys. \textbf{44}, 5003--5020 (2003).}

\bibitem{HeinosaariZimanBook}
T. Heinosaari and M. Ziman, \textit{The Mathematical Language of Quantum
Theory} \href{https://doi.org/10.1017/CBO9781139031103}{(Cambridge University Press, Cambridge, 2012).}

\bibitem{Ozawa84}
M. Ozawa, Quantum measuring processes of continuous observables, \href{https://doi.org/10.1063/1.526000}{J. Math. Phys. \textbf{25}, 79--87 (1984).}

\bibitem{ChiribellaD'ArianoPerinotti09}
G. Chiribella, G. M. D’Ariano, and P. Perinotti, Realization schemes for quantum instruments in finite dimensions, \href{https://doi.org/10.1063/1.3105923}{J. Math. Phys. \textbf{50}, 042101 (2009).}

\bibitem{Pellonpaa13}
J.-P. Pellonpää, Quantum instruments: II. Measurement theory, \href{https://doi.org/10.1088/1751-8113/46/2/025303}{J. Phys. A: Math. Theor. \textbf{46}, 025303 (2013).}

\bibitem{HorodeckiShorRuskai03}
M. Horodecki, P.W. Shor, and M.B. Ruskai, Entanglement breaking channels, \href{https://doi.org/10.1142/S0129055X03001709}{Rev. Math. Phys. \textbf{15}, 629--641 (2003).}

\bibitem{Choi75}
M. Choi, Completely positive linear maps on complex matrices, \href{https://doi.org/10.1016/0024-3795(75)90075-0}{Linear Alg. Appl. \textbf{10}, 285--290 (1975).}

\bibitem{MartensMuynck90}
H. Martens and W. de Muynck, Nonideal quantum measurements, \href{https://doi.org/10.1007/BF00731693}{Found. Phys. \textbf{20}, 255--281 (1990).}

\end{thebibliography}
\end{document}